%% file: 0-main.tex
\newif\iflncs
\lncsfalse

\iflncs
\documentclass[envcountsect,envcountsame]{llncs}
\makeatletter
\renewcommand\subsubsection{\@startsection{subsubsection}{3}{\z@}{-18\p@ \@plus -4\p@ \@minus -4\p@}{0.5em \@plus 0.22em \@minus 0.1em}{\normalfont\normalsize\bfseries\boldmath}}
\makeatother
\else
\documentclass[11pt]{article}
\fi

\newif\ifnotes
\notesfalse

\iflncs
\else
\def\showtableofcontents{1}
\fi 

\input{template}
\input{macros}

\title{Quantum Advantage from Any Non-Local Game} 

\iflncs
\author{Anonymous Submission to CRYPTO 2022}
\institute{}
\else
\author{Yael Kalai\thanks{MIT and Microsoft Research. E-mail: \texttt{yael@microsoft.com}.} \and Alex Lombardi\thanks{MIT. E-mail: \texttt{alexjl@mit.edu}. Research supported by a Charles M. Vest Grand Challenges Fellowship, an NDSEG Fellowship, and grants of the third author.} \and Vinod Vaikuntanathan\thanks{MIT. E-mail: \texttt{vinodv@mit.edu}. Research supported in part by DARPA under Agreement No. HR00112020023, a grant from the MIT-IBM Watson AI, a grant from Analog Devices and a Microsoft Trustworthy AI grant. Any opinions,
findings and conclusions or recommendations expressed in this material are those of the author(s) and do not
necessarily reflect the views of the United States Government or DARPA.} \and Lisa Yang\thanks{MIT. E-mail: \texttt{lisayang@mit.edu}. Research supported by an NSF Graduate Research fellowship and grants of the third author. \vinod{Lisa: please add ack.}}}
\date{\today}
\fi

\begin{document}
\maketitle

\begin{abstract} 
We show a general method of compiling any $k$-prover non-local game into a single-prover interactive game maintaining the same (quantum) completeness and (classical) soundness guarantees (up to negligible additive factors in a security parameter). Our compiler uses any quantum homomorphic encryption scheme (Mahadev, FOCS 2018; Brakerski, CRYPTO 2018) satisfying a natural form of correctness with respect to auxiliary (quantum) input. The homomorphic encryption scheme is used as a cryptographic mechanism to simulate the effect of spatial separation, and is required to evaluate $k-1$ prover strategies (out of $k$) on encrypted queries. 

In conjunction with the rich literature on (entangled) multi-prover non-local games starting from the celebrated CHSH game (Clauser, Horne, Shimonyi and Holt, Physical Review Letters 1969), our compiler gives a broad framework for constructing mechanisms to classically verify quantum advantage.
\end{abstract}

\vfill


%

\iflncs
\else
\ifnum\showtableofcontents=1
{
\thispagestyle{empty}
\newpage
\setcounter{tocdepth}{2}
\tableofcontents
\thispagestyle{empty}
 }
\clearpage
\fi
\fi 

\input{1-intro.tex}

\input{2-prelims}
\input{3-game-compiler}
\input{4-quantum-advantage}

\section*{Acknowledgements}
We thank Fermi Ma for comments regarding QHE correctness on an earlier draft of this paper.

\bibliographystyle{alpha}
\bibliography{quantum,crypto,abbrev3}

\input{appendix.tex}

\end{document}

%% file: template.tex
\usepackage[utf8]{inputenc}
\usepackage{beton}

\usepackage{amsmath}
\usepackage{amssymb}
\usepackage{amsthm}
\usepackage[dvipsnames]{xcolor}
\usepackage{mathtools}
\usepackage[utf8]{inputenc}
\usepackage{amsfonts}
\usepackage{graphicx}
\usepackage{mathrsfs}
\usepackage{physics}
\usepackage{soul}
\usepackage{bm} 
\usepackage[ruled,linesnumbered]{algorithm2e}
\usepackage[T1]{fontenc}

\iflncs
\else
\usepackage{fullpage}
\fi

\definecolor{DarkBlue}{RGB}{0,0,150}
\definecolor{DarkRed}{RGB}{150,0,0}
\definecolor{DarkGreen}{RGB}{0,150,0}
\usepackage[colorlinks,linkcolor=Maroon,citecolor=DarkBlue]{hyperref}
\usepackage[capitalize,nameinlink]{cleveref}
  \crefname{step}{Step}{Steps}
\usepackage[hyperpageref]{backref}

\newcommand{\authnote}[3]{\textcolor{#3}{[{\footnotesize {\bf #1:} { {#2}}}]}}

\newcommand{\alex}[1]{\ifnotes \authnote{Alex}{#1}{WildStrawberry} \fi}
\newcommand{\yael}[1]{\ifnotes \authnote{Yael}{#1}{Plum} \fi}

\newcommand{\vinod}[1]{\ifnotes \authnote{Vinod}{#1}{RedViolet} \fi}

\newcommand{\lisa}[1]{\ifnotes \authnote{Lisa}{#1}{BrickRed} \fi}

%% file: macros.tex
\iflncs
\else
\newtheorem{theorem}{Theorem}[section]

\newtheorem{claim}[theorem]{Claim}

\newtheorem{corollary}[theorem]{Corollary}

\newtheorem{definition}[theorem]{Definition}

\newtheorem{remark}[theorem]{Remark}

\fi

\Crefname{importedtheorem}{Imported Theorem}{Imported Theorems}
\Crefname{theorem}{Theorem}{Theorems}
\Crefname{proposition}{Proposition}{Propositions}
\Crefname{claim}{Claim}{Claims}
\Crefname{lemma}{Lemma}{Lemmas}
\Crefname{conjecture}{Conjecture}{Conjectures}
\Crefname{corollary}{Corollary}{Corollaries}
\Crefname{construction}{Construction}{Constructions}
\Crefname{property}{Property}{Properties}

\theoremstyle{definition}

\Crefname{definition}{Definition}{Definitions}
\Crefname{assumption}{Assumption}{Assumptions}
\Crefname{notation}{Notation}{Notations}

\theoremstyle{remark}

\Crefname{question}{Question}{Questions}
\Crefname{remark}{Remark}{Remarks}
\Crefname{comment}{Comment}{Comments}
\Crefname{fact}{Fact}{Facts}

\newcommand{\Gen}{\mathsf{Gen}}
\def\Eval{\mathsf{Eval}}

\newcommand{\pST}{\; \middle| \;}

\newcommand{\q}{q}
\newcommand{\Nat}{\mathbb{N}} 
\newcommand{\zo}{\{0,1\}}
\newcommand{\U}{U} 
\newcommand{\C}{C} 
\renewcommand{\a}{a}
\renewcommand{\v}{v}

\newcommand{\good}{\mathsf{Good}}
\newcommand{\bad}{\mathsf{Bad}}
\newcommand{\CS}{\mathsf{CS}}
\newcommand{\IG}{Interactive Game}
\newcommand{\ig}{interactive game}
\newcommand{\BQP}{\mathsf{BQP}}
\newcommand{\ppt}{\mathsf{PPT}}
\newcommand{\PPT}{\mathsf{PPT}}
\newcommand{\QPT}{\mathsf{QPT}}
\newcommand{\Enc}{\mathsf{Enc}}
\newcommand{\Dec}{\mathsf{Dec}}

\newcommand{\ct}{\mathsf{ct}}

\newcommand{\QHE}{\mathsf{QHE}}
\newcommand{\QFHE}{\mathsf{QFHE}}
\newcommand{\LQHE}{\mathsf{QHE}}
\newcommand{\bQHE}{\mathsf{bQHE}}

\def\cA{{\cal A}}
\def\cB{{\cal B}}
\def\cC{{\cal C}}
\def\cD{{\cal D}}

\def\cF{{\cal F}}
\def\cG{{\cal G}}
\def\cH{{\cal H}}
\def\cI{{\cal I}}

\def\cP{{\cal P}}
\def\cQ{{\cal Q}}
\def\cR{{\cal R}}

\def\cT{{\cal T}}

\def\cV{{\cal V}}

\newcommand{\RegA}{\mathcal{A}}
\newcommand{\RegB}{\mathcal{B}}

\newcommand{\RegH}{\mathcal{H}}

\newcommand{\secp}{\lambda}
\newcommand{\poly}{\mathsf{poly}}
\newcommand{\negl}{\mathsf{negl}}

\newcommand{\sk}{\mathsf{sk}}

\newcommand{\Hyb}{\mathsf{Hyb}}

\newcommand{\Id}{\mathsf{Id}}

\newcommand{\brho}{\bm{\rho}}

\newcommand{\tensor}{\otimes}


%% file: 1-intro.tex
\section{Introduction}
\label{sec:intro}

Quantum computing promises to usher in a new era in computer science. There are several ongoing large-scale efforts by Google, Microsoft, IBM and a number of startups to demonstrate {\em quantum advantage}, that is, to construct a quantum machine capable of computations that no classical machine can replicate with comparable speed. Central to the demonstration of quantum advantage is the question of verification. That is, can a classical machine\footnote{One could relax this requirement to allow a verifier with {\em limited} quantum capabilities, for example, preparing and transmitting single-qubit states~\cite{aharonov2008interactive,fitz-kashefi}. In this paper, we deal exclusively with classical verifiers.} {\em check} that a given device is performing computations that no classical device is capable of?  

Such verification schemes for quantum advantage are central to the ongoing race to construct non-trivial quantum computers.  Several mechanisms have been proposed to classically verify quantum advantage:

\begin{enumerate}
\item \textbf{The Algorithmic Method}: Ask the quantum machine to do a computation, such as factor a large number~\cite{Shor97}, that is believed to be (sub-exponentially) hard for classical machines~\cite{LenstraLMP90}. This is the most straightforward way to check quantum advantage, yet it is concretely very expensive and requires millions of (high-fidelity) qubits which is likely infeasible for near-term quantum machines (see \cite{gidney-ekera} for a recent, optimized estimate).

\item \textbf{The Sampling Method:} Ask the quantum machine to (approximately) sample from a distribution, which is presumably hard to accomplish classically~\cite{BJS10,aaronson2010computational}. This results in tasks that are feasible for near-term quantum machines~\cite{sycamore}, yet suffers from two problems. First, checking that the output is correct (at least in the current proposals) is an exponential-time computation that quickly becomes infeasible for a classical machine. Secondly, the question of whether approximate sampling is classically hard is less clear and has come under question~\cite{sunway}, despite hardness results for exact sampling~\cite{BFNV18}.

\item \textbf{The Interactive Method:}  The most recent example of a way to demonstrate quantum advantage relies on using interactive proofs and the fact that quantum algorithms cannot necessarily be ``rewound'' which, in turn, is a consequence of the {\em quantum no-cloning principle}. Starting with the work of Brakerski, Christiano, Mahadev, Vazirani and Vidick~\cite{BCMVV18}, there have been a handful of proposals reducing the number of rounds~\cite{BKVV20}, improving computational efficiency and near-term feasibility~\cite{KCVY21}, and extending to more powerful guarantees such as self-testing~\cite{metger-vidick} and quantum delegation \cite{FOCS:Mahadev18a}. The practicality of the quantum machine in these protocols (eg. \cite{zhu2021interactive}) seems to fall somewhere between (1) and (2), i.e., the prover is more practical than with the algorithmic method but less than the sampling method.  
Compared to (2), verification by the classical machine is efficient.  
\end{enumerate}

\noindent
Meanwhile, from more than half a century ago, starting with the celebrated works of John Bell~\cite{Bell} and Clauser, Horne, Shimony and Holt~\cite{CHSH} to explain the Einstein-Podolsky-Rosen (EPR) paradox, there is a rich body of work that came up with several methods to test that two or more spatially separated, non-communicating quantum devices can generate (testable) correlations that no pair of classical devices can produce. This immediately gives us a test of quantum advantage of a very different kind, one whose validity relies on the (non-falsifiable) physical assumption that devices do not communicate.

\begin{enumerate}
\setcounter{enumi}{3}
\item \textbf{The Non-Local Method:}  Here, quantum advantage comes from {\em entanglement}. Ask for a pair (or more) of spatially separated, but entangled, quantum devices to perform a computation that no classical machine can replicate without communication.  A classical example is the famed CHSH game~\cite{CHSH}, or other violations of the Bell inequality. A rich class of non-local games have since been constructed and this is an active area of current research with many breakthroughs in the last few years (e.g. \cite{PhysRevLett.65.3373,Aravind:01,PhysRevLett.71.1665,greenberger2007going,ReichardtUV13} up until \cite{ji2020mipre}). While many of these games are attractive and super-efficient (e.g. using only a small constant number of qubits), this suffers from possible ``loopholes''. That is, {\em how does one ensure that the provers {\em do not} subliminally communicate in some way}?  Indeed, constructing so-called loophole-free demonstrations of Bell inequality violations is an active area of research~\cite{loopholefree}.  
\end{enumerate}

\subsection{Our Results}

A tantalizing question that comes out of the discussion above is the relation between techniques used to construct {\em multi-prover} non-local games and the ones used to construct \emph{single-prover} proofs of quantumness. On the one hand, the former relies on entanglement whereas the latter relies on quantum non-rewinding. On the other hand, the intellectual foundations of single-prover protocols for quantum advantage sometimes come from non-local games. Indeed, Kahanamoku-Meyer, Choi, Vazirani and Yao~\cite{KCVY21}, in a recent work, call their interactive single-prover protocol for quantum advantage a ``computational Bell test''. Yet, there have been no formal connections between multi-prover non-local games and single-prover interactive protocols for quantum advantage. 

The central result of this work is the construction of a {\em simple} and {\em general} compiler that gives us an interactive protocol demonstrating quantum advantage starting from two building blocks.

\begin{theorem}\label{thm:intro:main}
Given
\begin{itemize}
    \item {\em Any} $k$-prover non-local game quantum value $c$ and classical value $s$; and 
    \item {\em Any} quantum homomorphic encryption scheme satisfying correctness with respect to auxiliary input (Definition~\ref{def:QHE-aux}), that can implement homomorphic evaluation of {\em at least $k-1$} provers in the non-local game.
\end{itemize}
there is a $2k$-round\footnote{We refer to a message from the prover to the verifier or the verifier to the prover as a round.} single-prover interactive game with completeness $c$, realized by a quantum polynomial-time algorithm, and soundness $s + \negl(\secp)$, against any classical polynomial-time algorithm. Here, $\secp$ is a security parameter that governs the security of the quantum homomorphic encryption scheme.
\end{theorem}
This gives us a rich class of protocols by mixing and matching any non-local game (CHSH, Magic Square, Odd-Cycle Test, GHZ, and many more) and any powerful enough quantum homomorphic encryption scheme (\cite{Mah18a,Bra18} both satisfy the appropriate correctness property; see \cref{app:locality}).

For example, an instantiation of our method using the CHSH game and the quantum FHE scheme of Mahadev~\cite{Mah18a} gives us a protocol that already improves on the state of the art in several respects:

\begin{itemize}
    \item Compared to \cite{BCMVV18}, the protocols that come out of our framework do not require the adaptive hardcore bit property, and therefore can be instantiated from a variety of assumptions including Ring Learning with Errors (Ring-LWE), resulting in better efficiency.
    \item Compared to \cite{BKVV20}, our protocols do not require random oracles. \cite{BKVV20} uses random oracles in a crucial way to obtain an instantiation from Ring-LWE, bypassing the adaptive hardcore bit property. They simultaneously achieve a two-round protocol in contrast to the $4$-round protocol of \cite{BCMVV18}. 
    \item Compared to \cite{KCVY21}, our protocols require 4 rounds as opposed to their 6. 
\end{itemize}


\iflncs
\medskip\noindent
\fi
Our framework elucidates interactive protocols for quantum advantage by separating them into an information-theoretic component (namely, the non-local game) and a cryptographic component (namely, the QFHE). We believe that the richness and extreme simplicity of our transformation could eventually lead to a protocol that concretely improves on the family of protocols in (3) above in a significant way. Philosophically, a remarkable aspect of our framework is that it uses cryptography as a bridge to connect two apparently different sources of quantumness, namely {\em entanglement} (used to construct non-local games) and {\em non-rewinding} (which underlies the reason why single-prover proofs of quantumness as in \cite{BCMVV18} work).

\iflncs
\medskip\noindent
\fi 
We proceed to describe a self-contained outline of our transformation below. 

\subsection{Technical Overview}

Our framework is remarkable in its simplicity, drawing  inspiration from similar compilers in the classical world~\cite{BMW98,C:BitChi12,KalaiRR13,KalaiRR14}.
For the rest of the introduction, we focus on compiling {\em two-prover} non-local games and  describe how to extend our construction to $k$ provers in Section~\ref{sec:k-provers}. 
The basic idea is to convert any non-local game into a single prover protocol by sending the prover the queries sent to \emph{both} provers in the non-local game.  It is easy to see that if we do this ``in the clear," we lose soundness, since by giving both queries to a single prover, we lose the spatial separation condition. This is precisely where the encryption scheme comes into play.

\paragraph{Using Cryptography to Enforce Spatial Separation.}

We enforce spatial separation using an encryption scheme. Specifically, in our protocol the verifier samples two queries $q_1, q_2$ by emulating the verifier in the non-local game, and the interactive protocol proceeds in rounds, as follows:
\begin{enumerate}
    \item The verifier sends $\Enc(q_1)$ to the prover. 
    \item The prover sends an encrypted answer $\Enc(a_1)$.
    \item The verifier then sends the second query $q_2$ (in the clear).
    \item The prover replies with an answer $a_2$. 
\end{enumerate} 
The verifier accepts if the decrypted transcript would be accepted in the non-local game.  

\paragraph{Completeness.}  We use the properties of the underlying encryption scheme to argue that any strategy of the (quantum) provers in the non-local game can be emulated in our protocol (partially ``under the hood" of the encryption scheme). For this, we need our homomorphic encryption scheme to have the following two properties:

\begin{itemize}
    \item It needs to support the quantum operations done by the first prover in the non-local game.
    \item It needs to be {\em entanglement preserving} in the following sense: if a quantum circuit $C$ is defined using a register $\RegA$
    ~and is homomorphically evaluated using a state $\ket{\Psi}_{\RegA\RegB}$ that has entanglement between $\RegA$ and another register $\RegB$ (which is not input to $C$), the decrypted output $y$ should remain correlated with $\RegB$ (as it would be if $C$ were evaluated using the $\RegA$ register of $\ket{\Psi}$). This property can be enforced via a natural ``aux-input correctness'' property (\cref{def:QHE-aux}).  
\end{itemize}  
We view correctness with respect to auxiliary input as a natural requirement that one should expect to hold for a ``typical'' $\QHE$ scheme. In \cref{app:locality}, we show that it follows from ``plain'' correctness under mild structural assumptions; namely, if decryption is ``bit-by-bit'' and homomorphic evaluation is ``local.'' In particular, this implies that seminal homomorphic encryption schemes due to Mahadev and Brakerski~\cite{Mah18a,Bra18} satisfy these two properties.

\begin{remark}
\emph{We note that a 2-message version of this protocol was used by Kalai, Raz and Rothblum~\cite{KalaiRR13,KalaiRR14} (and in followup works) to obtain a 2-message delegation scheme for classical computations. In this 2-message variant, the verifier simply sends all of the queries in the first message each encrypted under independently chosen keys, and the prover computes each answer homomorphically. This 2-message variant has the desired completeness property described above.}

\emph{
However, it  {\em cannot} be used as a protocol that has a gap between classical and quantum winning probabilities.  This is the case since the soundness of the resulting 2-message delegation scheme is proven only if the original non-local game has {\em non-signaling soundness}, namely, assuming the provers cannot cheat via a non-signaling (and  therefore, a quantum) strategy. In other words, classical provers may be able to implement non-signaling strategies, and these strategies can have value which is at least as high as the quantum value.
~Thus, we cannot argue that in the \cite{KalaiRR14}-like 2-message protocol, a classical cheating prover cannot emulate this non-signaling strategy. Instead, we consider the 4-message variant described above.} 

\emph{This 4-message protocol bears some resemblance to (although is significantly simpler than) the classical succinct argument system of \cite{C:BitChi12}, which is also concerned with compiling MIPs into single-prover protocols. }
\end{remark}

\paragraph{Soundness.} We argue that any computationally bounded {\em classical} cheating prover in the 4-message protocol can be converted into a {\em local} prover strategy in the 2-player game, with roughly the same acceptance probability (see Theorem~\ref{thm:main}).  
This allows us to argue quantum advantage: Namely, an honest quantum prover can obtain the (quantum) value of the 2-player game, whereas any classical prover can obtain only the classical value of the 2-player game.

To prove our desired soundness condition, we rewind the (classical) prover, something that cannot be done with a quantum prover, nor can be done in the 2-message setting. 
Specifically, fix a classical  cheating prover $P^*$ for the 4-message protocol, and assume without loss of generality that it is deterministic (otherwise, fix its random coins to ones that maximize the success probability). The corresponding local provers $(P^*_1,P^*_2)$ for the 2-player game are defined as follows:
\begin{itemize}
\item Fix the first 2 messages in the 4-message protocol, by choosing any $q'_1$ (e.g., it can be the all-zero string), computing $\ct_1=\Enc(q'_1)$ and emulating $P^*$ to compute $\ct_2=P^*(\ct_1)$. 
\item $P^*_2$ has $(\ct_1,\ct_2)$ hardwired.\footnote{We note that it suffices to hardwire $\ct_1$ since $\ct_2=P^*(\ct_1)$ can be computed from $\ct_1$.}  On input a  query~$q_2$, it emulates the response of $P^*$, conditioned on the first two messages being $(\ct_1,\ct_2)$, to obtain~$a_2=P^*(\ct_1,\ct_2,q_2)$.  It outputs~$a_2$.
\item $P^*_1$ has the truth-table of $P^*_2$ hardwired.  On input $q_1$ it computes for every possible $a_1$ the probability (over $q_2$ sampled conditioned on $q_1$) that the verifier accepts $(q_1, a_1, q_2, P^*_2(q_2))$ and sends $a_1$ with the maximal acceptance probability. 
\end{itemize} 

To argue that $(P^*_1,P^*_2)$ convinces the verifier to accept with essentially the same probability that~$P^*$ does, we rely on the security of the encryption scheme.  
Namely, we argue (by contradiction) that if the verifier accepts with probability significantly smaller than the acceptance probability of $P^*$ in the 4-message protocol, then there exists an adversary~$\cA$ that breaks semantic security of the underlying encryption scheme.  

Specifically, $\cA$ is given a challenge ciphertext~$\ct$, which it will use to define $P^*_1$ and $P^*_2$ (as above), and will use $P^*_1$ and $P^*_2$ in his attack.  The first barrier is that $P^*_1$ may not be efficient, which results with $\cA$ being inefficient.  Thus, we first argue that (a good enough approximation of) $P^*_1$ can be emulated in time $2^{|q_1|+|a_1|}\cdot\poly(\secp)$, which results in $\cA$ running in that time as well. Loosely speaking, this is done by estimating for every $a_1$, the probability of acceptance, and taking $a_1$ with the maximal acceptance probability.  This explains the exponential blowup in~$a_1$.  To estimate the probability that $(q_1,a_1)$ is accepted (w.r.t.~$P^*_2$) we do an empirical estimation by choosing many $q_2$'s from the residual query distribution (conditioned on~$q_1$) and compute the fraction of $q_2$'s for which $V(q_1,q_2,a_1,P^*(q_2))=1$. We use the Chernoff bound to argue that this empirical estimation is close to the real probability of acceptance. Note that this estimation requires  sampling $q_2$ from the residual query distribution (conditioned on~$q_1$), and it is not clear that this can be done efficiently.  Thus, we hardwire into $\cA$ many such samples for each and every possible $q_1$, which results with the exponential blowup in~$q_1$.\footnote{However, this blowup can be avoided for many non-local games; for example, if $q_1$ and $q_2$ are independent, or if the game is a parallel repetition of a constant-size game. For the latter case, see \cref{subsec:poly-assumptions} for more details.}

Given this $ 2^{|q_1| + |a_1|}\poly(\secp)$-time implementation of $P^*_1$, the hardness reduction proceeds as follows. Loosely speaking, given $q_1,q_1',\ct$ where $\ct$ is either $\Enc(q_1)$  or $\Enc(q_1')$ with equal probability, 
 we run the 2-prover game with $(P^*_1,P^*_2)$ corresponding to $\ct$ and with the first query being $q_1$, and then run it again with the first query being $q_1'$.  If the verifier accepts exactly one of these executions, then we guess $b$ to be the one corresponding to the winning execution; otherwise, we output a random guess for $b$.

 Note that in our attack we broke the security in time $2^{|q_1|+|a_1|}\cdot\poly(\secp)$ and thus need to assume that the encryption scheme has that level of security. For constant-size games, this is equivalent to polynomial security; in general, this requires assuming the sub-exponential security of the underlying FHE. However, as discussed in \cref{subsec:poly-assumptions}, it is possible to rely on polynomially-secure FHE for certain superconstant-size non-local games such as parallel repetitions of a constant-size game. We refer the reader to Section~\ref{sec:compiler} for the details, and to the proof of Theorem~\ref{thm:main} for the formal analysis.

%% file: 2-prelims.tex
\section{Preliminaries}

\def\PT{\mathsf{PT}}

We let $\PT$ denote deterministic polynomial time,  $\PPT$ denote probabilistic polynomial time and $\QPT$ denote quantum polynomial time.
For any random variables $A$ and $B$ (possibly parametrized by a security parameter $\secp$), we use the notation $A\equiv B$ to denote that $A$ and $B$ are identically distributed, and $A\approx_s B$ to denote that $A$ and $B$ have negligible ($\secp^{-\omega(1)}$) statistical distance. We use $A \approx_c B$ to denote that $A$ and $B$ are computationally indistinguishable, namely for every $\PPT$ distinguisher $D$, $|\Pr[D(A) = 1]-\Pr[D(B)=1]| = \secp^{-\omega(1)}$. 

\paragraph{Quantum States.} Let $\cH$ be a Hilbert space of finite dimension~$2^n$ (thus, $\cH\simeq \mathbb{C}^{2^n}$).
A (pure, $n$-qubit) quantum state $\ket{\Psi}\in \cH$ is an element of the form
$$\ket{\Psi}=\sum_{b_1,\ldots,b_n\in\{0,1\}} \alpha_{b_1,\ldots,b_n}\ket{b_1,\ldots,b_n}
$$
where $\{\ket{b_1,\ldots,b_n}\}_{b_1,\ldots,b_n\in\{0,1\}}$ forms an orthonormal basis of $\cH$, $\alpha_{b_1,\ldots,b_n}\in \mathbb{C}$ and 
$$
\sum _{b_1,\ldots,b_n\in\{0,1\}}|\alpha_{b_1,\ldots,b_n}|^2=1.
$$
We refer to $n$ as the number of registers of $\ket{\Psi}$. 

We denote by $M(\ket{\Psi})$ the outcome of measuring $\ket{\Psi}$ in the standard basis (throughout this work we are only concerned with standard basis measurements).  
For any set of registers $\cI\subseteq[n]$, we denote by $M_\cI(\ket{\Psi})$ the outcome of measuring only the $\cI$ registers of $\ket{\Psi}$ in the standard basis.

A \emph{mixed state} $\brho$ over $\cH$ is a density operator (we use $\mathbf{S}(\cH)$ to denote the space of density operators on $\cH$) normalized so that $\Tr(\brho) = 1$. Every pure state $\ket{\Psi}$ has a corresponding rank $1$ mixed state $\ketbra{\Psi}$ such that for any PSD projection $\Pi$, we have $||\Pi \ket{\Psi}||^2 = \Tr(\Pi \ketbra{\Psi})$. 

We sometimes divide the registers of $\RegH$ into named registers, denoted by calligraphic upper-case letters, such as $\cA$ and $\cB$, in which case we also decompose the Hilbert space into $\cH = \cH_{\cA} \otimes \cH_\cB$, so that each pure quantum state $\ket{\Psi}$ is a linear combination of quantum states $\ket{\Psi_\cA}\otimes\ket{\Psi_\cB} \in \cH_\cA\otimes\cH_\cB$.  

For any mixed state $\brho$, we denote by $\brho_{\RegA}$ the reduced density operator
\[\brho_{\RegA}=\Tr_\cB(\brho)\in \mathbf{S}(\cH_\cA),
\]
where $\Tr_\cB$ is the ``partial trace'' linear operator defined by 

\[ \Tr_\cB(\ket{\Psi_\cA}\bra{\Psi_\cA}\otimes \ket{\Psi_\cB}\bra{\Psi_\cB})=\ket{\Psi_\cA}\bra{\Psi_\cA}\cdot \Tr(\ket{\Psi_\cB}\bra{\Psi_\cB}).
\]

\alex{define trace distance}

\paragraph{Quantum Ciruits and Locality.} A quantum circuit is a sequence of elementary quantum gates (taken from some complete basis) and measurement operations. These operations are applied to an initial state $\ket{\psi}$ and result in some final state $\ket{\psi'}$. For further definitions related to quantum circuits and gate types, we refer the reader to~\cite{NC00}. 

One direct consequence of this model is a form of \textbf{locality}: if $\ket{\psi} = \ket{\psi}_{\RegA\RegB}$ is shared between two registers ($\RegA$ and $\RegB$), operations of the form $C\tensor \Id_{\RegB}$ (where $C$ is a quantum circuit acting only on $\RegA$) can be performed \emph{without possession of the register $\RegB$}.

\subsection{Non-Local Games}

\begin{definition}[$k$-Player Non-local Game]
A $k$-player non-local game $\cG$ consists of a $\ppt$ sampleable query distribution $\cQ$ over $(\q_1,\hdots, q_k)\in\left(\zo^{n}\right)^k$ and a polynomial time verification predicate $\cV(\q_1,\q_2, \hdots, q_k, \a_1,\a_2, \hdots, a_k)\in\zo$, where each $a_i \in \zo^{m}$. The classical value and the quantum value of $\cG$ are defined below.

\begin{itemize}
    \item \textbf{Classical value:} The classical value $v$ of $\cG$ is defined as:
    \[\max_{P_1,\hdots, P_k:\{0,1\}^n \to \{0,1\}^m} \underset{(q_1, \hdots, q_k)\gets \cQ}{\Pr}\left[\cV(\q_1,\hdots, q_k,P_1(\q_1),\hdots, P_k(q_k))=1\right]
    \]
    
    \item \textbf{Quantum (entangled) value:} The quantum value $\v^*$ is defined as:  
    \[\max_{\substack{\ket{\Psi}\in \RegH_1 \tensor \hdots \tensor \RegH_k\\ U_i\in U(\RegH_i \tensor \mathbb C^n), 1\leq i\leq k}}
    \underset{(q_1, \hdots, q_k)\gets \cQ}{\Pr}\left[\cV(\q_1,\hdots, q_k,a_1,\hdots, a_k)=1\right]
    \]
    where the maximum is over all $k$-partite states $\ket{\Psi}\in \RegH_1\tensor \hdots \tensor \RegH_k$ (each $\cH_i$ is an arbitrary finite-dimensional Hilbert space), unitaries $U_i$ acting on $\log \dim_{\mathbb C} \RegH_i + n$ qubits respectively, and answers $a_i$ computed by applying $U_1 \tensor \hdots \tensor U_k$ to $\ket{\Psi}\tensor \ket{q_1\hdots q_k}$, and measuring the first $m$ qubits in each $\RegH_i$; that is, 
    
    \[(a_1,a_2, \hdots, a_k)\gets M_{\cI}(U_1\tensor \hdots \tensor U_k \left(\ket{\Psi}\tensor \ket{q_1, \hdots, q_k} \right))
    \]
    for $\cI = \{1, \hdots, m\} \times [k]$. We remark that without loss of generality, $\ket{\Psi}$ can be (and is above) taken to be a pure state. 
\end{itemize}
\end{definition}
We are interested in non-local games where quantum strategies can win with probability strictly more than any classical strategy, that is ones for which $v^* > v$.

\begin{remark}
\emph{In this work, it is crucial to consider the complexity of the honest provers in a non-local game, which is often not a parameter of interest in the literature but instead hidden in the description of the provers as a tuple of unitaries $(U_1,\hdots, U_k)$ and state $\ket{\Psi}$.}

\emph{We explicitly consider non-local games where the each prover's unitary $U_i$ can be implemented as a quantum circuit $C_i$ of size polynomial in a security parameter $\secp$. When $\cG$ has constant size (such as the CHSH or Magic square games) this holds automatically, but this is a non-trivial requirement when the size of $\cG$ grows with $\secp$.}
\end{remark}

\subsection{Quantum Homomorphic Encryption} \label{sec:QHE}

We define the notion of quantum homomorphic encryption~\cite{Mah18a,Bra18} that is central to our framework. Unlike in \cite{Bra18}, our definition requires a form of \emph{correctness with respect to auxiliary input}. However, we show that this definition holds for QHE schemes satisfying mild additional requirements, and in particular holds for the \cite{Mah18a,Bra18} schemes. 

\alex{I commented out the ``bit-by-bit'' definition for simplicity -- we can use it in the appendix if it helps there.}

\begin{definition}[Quantum Homomorphic Encryption (QHE)]\label{def:QHE-aux}
A quantum homomorphic encryption scheme $\QHE=(\Gen,\Enc,\Eval,\Dec)$ for a class of quantum circuits $\cC$ is a tuple of algorithms with the following syntax:
\begin{itemize}
    \item $\Gen$ is a $\PPT$ algorithm that takes as input the security parameter $1^\secp$ and outputs a (classical) secret key $\sk$ of $\poly(\secp)$ bits;
    \item $\Enc$ is a $\PPT$ algorithm that takes as input a secret key $\sk$ and a classical input $x$, and outputs a ciphertext $\ct$;
    \item $\Eval$ is a $\QPT$ algorithm that takes as input a tuple $(\C,\ket{\Psi},\ct_{\mathrm{in}})$, where $\C:\cH\times(\mathbb{C}^2)^{\tensor n}\rightarrow (\mathbb{C}^2)^{\tensor m}$ is a quantum circuit, $\ket{\Psi}\in\cH$ is a quantum state, and $\ct_{\mathrm{in}}$ is a ciphertext corresponding to an $n$-bit plaintext. 
    $\Eval$ computes a quantum circuit
    $\Eval_C(\ket{\Psi}\tensor \ket{0}^{\poly(\secp, n)},\ct_{\mathrm{in}})$ which outputs a ciphertext $\ct_{\mathrm{out}}$. If $C$ has classical output, we require that $\Eval_C$ also has classical output.

    \item $\Dec$ is a $\PT$ algorithm that takes as input a secret key $\sk$ and ciphertext $\ct$, and outputs a state $\ket{\phi}$. Additionally, if $\ct$ is a classical ciphertext, the decryption algorithm outputs a classical string $y$.
\end{itemize}

The above syntax is more general than the form required in \cite{Bra18}; we elaborate on this difference below. We require the following two properties from $(\Gen,\Enc,\Eval,\Dec)$:
\begin{itemize}
    \item \textbf{Correctness with Auxiliary Input:} For every security parameter $\secp\in\Nat$, any quantum circuit $C:\cH_{\RegA} \times (\mathbb{C}^2)^{\tensor n} \to \{0,1\}^*$ (with classical output), any quantum state $\ket{\Psi}_{\RegA \RegB} \in\cH_{\RegA} \tensor \cH_{\RegB}$,  any message $x\in \{0,1\}^n$, any secret key $\sk \gets \Gen(1^\secp)$ and any ciphertext $\ct \gets \Enc(\sk,x)$, the following states have negligible trace distance:
    \begin{description}
        \item \textit{Game $1$.} Start with $(x, \ket{\Psi}_{\RegA \RegB})$. Evaluate $ C$ on $x$ and register $\RegA$, obtaining classical string $y$. Output $y$ and the contents of register $\RegB$.
        \item \textit{Game $2$.} Start with $\ct \gets \Enc(\sk, x)$ and $\ket{\Psi}_{\RegA \RegB}$. Compute  $\ct' \gets\Eval_C(\cdot \tensor \ket{0}^{\poly(\secp, n)} ,\ct)$ on register $\RegA$. Compute $y'= \Dec(\sk,\ct')$. Output $y'$ and the contents of register $\RegB$.
    \end{description}
    
    \item \textbf{$T$-Classical Security:} For any two messages $x_0, x_1$ and any $\poly(T(\secp))$-size classical circuit ensemble~$\cA$:
    \[\left|\Pr\left[\cA(\ct_0) = 1 \pST
        \begin{array}{l}
        \sk\gets\Gen(1^\secp)\\
        \ct_0\gets\Enc(\sk,x_0)\\
        \end{array}
        \right]
    -\Pr\left[\cA(\ct_1) = 1 \pST
        \begin{array}{l}
        \sk\gets\Gen(1^\secp)\\
        \ct_1\gets\Enc(\sk,x_1)\\
        \end{array}
        \right]\right|
    \le \negl(T(\secp))\enspace.
    \]
\end{itemize}
\end{definition}

In words, ``correctness with auxiliary input'' requires that if QHE evaluation is applied to a register $\RegA$ that is a part of a joint (entangled) state in $\cH_{\RegA}\tensor \cH_{\RegB}$, the entanglement between the QHE evaluated output and $\RegB$ is preserved.

\begin{remark}
\emph{A \emph{quantum fully homomorphic encryption ($\QFHE$)} is a $\QHE$ for the class of all poly-size quantum circuits. While \cite{Mah18a,Bra18} construct $\QFHE$ (with security against quantum distinguishers),
~weaker forms of $\QHE$ may yield more efficient quantum advantage protocols (see Section~\ref{sec:future-work} for discussion).}
\end{remark}

\begin{remark}
\emph{In our definition of security, we only
~consider classical attacks. Classical security is sufficient for the purposes of this work as  protocols for quantum advantage are required to have quantum completeness and classical soundness.}
\end{remark}

For the purposes of this paper, it suffices to know the following claim about the instantiability of \cref{def:QHE-aux}.

\begin{claim} \label{claim:QFHE}
The \cite{Mah18a,Bra18} QFHE schemes satisfy \cref{def:QHE-aux} with correctness holding for the class of all poly-size quantum circuits.
\end{claim}
\cref{claim:QFHE} can be verified by inspecting the constructions given in \cite{Mah18a,Bra18}. In \cref{app:locality}, we show mild generic conditions under which a QFHE scheme satisfies correctness with respect to auxiliary input, and sketch a proof of \cref{claim:QFHE} .

%% file: 3-game-compiler.tex
\section{Our Compiler: From Non-Local Games to Interactive Protocols}\label{sec:compiler}

In this section, we show how to use a quantum homomorphic encryption scheme satisfying aux-input correctness (\cref{def:QHE-aux}) to convert a 2-prover non-local game into a {\em single-prover} interactive protocol with {\em computational soundness}.

\begin{definition}
[Single-Prover Computationally Sound \IG]
A single-prover computationally sound ($\CS$) \ig~$\cG$ consists of an interactive $\ppt$ verifier $\cV$ that takes as input a security parameter $1^\secp$ and interacts with an interactive prover. The classical (computationally sound) value and the quantum (computationally sound) value of $\cG$ are defined below.
\begin{itemize}
    \item \textbf{Classical $\CS$ value:} $\cG$ has classical $\CS$ value $\geq \v$ if and only if  there exists an  interactive polynomial-size Turing machine~$\cP$ such that for every $\secp\in\mathbb{N}$,
    $$
    \Pr[(\cP,\cV)(1^\secp)=1]\geq v
    $$
    where the probability is over the random coin tosses of $\cV$, and where $(\cP,\cV)(1^\secp)\in\{0,1\}$ denotes the output bit of $\cV(1^\secp)$ after interacting with $\cP$.
    \item \textbf{Quantum $\CS$ value:} $\cG$ has quantum $\CS$ value $\geq \v^*$ if and only if there exists a Hilbert space $\cH$ and a quantum state $\ket{\Psi}\in \cH$ and an interactive $\QPT$ \yael{need to define ``interactive $\QPT$"}  prover $\cP$ such that for every $\secp\in\mathbb{N}$,
    $$
    \Pr[(\cP(\ket{\Psi}),\cV)(1^\secp)=1]\geq \v^*
    $$
    where the probability is over the randomness of $\cP$ and $\cV$.
    \end{itemize}
\end{definition}

\subsection{Our Transformation, $k=2$ case.} \label{sec:trans}
Fix a quantum homomorphic encryption scheme $\LQHE=(\Gen,\Enc,\Eval,\Dec)$ for a class of quantum circuits $\cC$ (e.g. \cite{Mah18a,Bra18}, see \cref{claim:QFHE}). We present a $\PPT$ transformation $\cT$ that converts any $2$-prover non-local game $\cG=(\cQ,\cV)$ into a single-prover computationally sound \ig~$\cT^\cG$ (asociated with security parameter $\secp$), defined as follows.
\begin{enumerate}
    \item The verifier samples $(\q_1,\q_2)\gets\cQ$, $\sk\gets\Gen(1^\secp)$, and $\hat{\q}_1\gets\Enc(\sk,\q_1)$. In the first round the verifier sends $\hat{\q}_1$ and in the third round he sends~$q_2$.
    \item The verifier, upon receiving $\hat{a}_1$ from the prover  in the first round, and $a_2$ in the second round, accepts if and only if $\cV(q_1,q_2,\Dec(\sk,\hat{a}_1),a_2)=1$.
\end{enumerate}

\begin{theorem}\label{thm:main}
Fix any $\LQHE$ scheme for a circuit class $\cC$, and any 2-player non-local game $\cG=(\cQ,\cV)$ with classical value $v$ and quantum value $\v^*$, such that the value $\v^*$ is obtained by a prover strategy $(C^*_1,C^*_2)$ with a quantum state $\ket{\Psi}\in\cH_\cA\otimes \cH_\cB$ where $C^*_1(\ket{\Psi}_\cA,\cdot)\in \cC$. Denote by $|q_1|$ and $|a_1|$ the lengths of the query and answer of the first prover, respectively.   If the underlying $\LQHE$ encryption scheme is $T$-secure, for  $T(\secp)=2^{|q_1|+|a_1|}\cdot\poly(\secp)$, then the following holds:
\begin{enumerate}
    \item The quantum $\CS$ value of $\cT^{\cG}$ is at least $\v^*$. 
    \item The classical $\CS$ value of $\cT^{\cG}$ is at most $\v+\negl(\secp)$.
\end{enumerate}
\end{theorem}

\begin{proof}
Fix any 2-player non-local game $\cG=(\cQ,\cV)$ with classical value~$v$ and quantum value $\v^*$, such that the value $\v^*$ is obtained by a prover strategy $(C^*_1,C^*_2)$ with the (joint) quantum state $\ket{\Psi}\in\cH_\cA\otimes\cH_\cB$, where $C^*_1\in \cC$.

\begin{enumerate}
    \item {\bf The quantum $\CS$ value of $\cT^{\cG}$ is at least $\v^*$.} Consider the following $\BQP$ prover $\cP^*$:
\begin{enumerate}
    \item In the first round, upon receiving $\hat{q}_1$, $\cP^*$ computes $\ct' \gets \Eval(\cdot_{\RegA},C^*_1,\hat{q}_1)$ on the register $\RegA$ of $\ket{\Psi}_{\RegA \RegB}$, and sends $\ct$. As internal state, $\cP^*$ retains the contents of register $\RegB$.
    \item In the second round, upon receiving $q_2$, $\cP^*$ uses its internal state $\brho_{\RegB}$ to compute and send $a'_2\gets C^*_2(\cdot_{\RegB},q_2)$. 
\end{enumerate}
We argue that 
$$
\Pr[\cV(q_1,q_2,a'_1,a'_2)=1]= \v^* - \negl(\secp),
$$
in the probability space where:
\begin{itemize}
    \item $(q_1,q_2)\gets \cQ$
    \item $\sk\gets \Gen(1^\secp)$ and $\hat{q}_1\gets \Enc(q_1,\sk)$
    \item $\ct' \gets \Eval(\cdot_{\RegA},C^*_1,\hat{q}_1)$ applied to $\ket{\Psi}_{\RegA \RegB}$ and $\brho$ set to the contents of $\RegB$.
    \item $a'_1=\Dec(\sk,\ct')$, and 
    \item $a'_2\gets C^*_2(\brho,q_2)$
\end{itemize}

By the aux-input correctness of $\LQHE$ (\cref{def:QHE-aux}) and the fact that $C^*_1\in\cC$, we see that for every $q_1$, the mixed state consisting of the distribution over $(a'_1, \brho)$ above is the same (up to negligible trace distance) as what would have been obtained by applying $C_1 \tensor \Id_{\RegB}$ to $(\ket{\Psi}_{\RegA \RegB}, q_1)$. By the contractivity of trace distance (with respect to the map defined by $C^*_2(\cdot_{\RegB}, q_2)$ and $\cV$), we conclude that
$$
\Pr[\cV(q_1,q_2,a'_1,a'_2)=1]=\Pr[\cV(q_1,q_2,a_1,a_2)=1] \pm \negl(\secp) = v^* \pm \negl(\secp),
$$
as desired.  

\item {\bf The classical $\CS$ value of $\cT^{\cG}$ is at most $\v+\negl(\secp)$.} 
Suppose for the sake of contradiction that the classical value of $\cT^\cG$ is $\v'=\v+\delta$ for a non-negligible $\delta=\delta(\secp)$. This implies that there exists a (deterministic) poly-size classical prover $\tilde{\cP}$ such that for every $\secp\in\mathbb{N}$,
$$
\Pr[(\tilde{\cP},\cV(\cT^{\cG}))(1^\secp)=1]=v',
$$
where $\cV(\cT^{\cG})$ denotes the verifier in the protocol $\cT^{\cG}$.
Next, for every $\secp\in\mathbb{N}$ we convert $\tilde{\cP}$ into (local) classical provers $(\cP_1,\cP_2)=(\cP_1(\secp),\cP_2(\secp))$ such that there exists a negligible function~$\mu$ such that for every $\secp\in\mathbb{N}$,
$$
\Pr[(\cP_1,\cP_2,\cV)=1]\geq \v' - \mu(\secp).
$$
Since $\v'-\mu(\lambda)>\v$ for sufficiently large $\secp$, this contradicts the fact that the classical value of $\cG$ is at most $\v$.

To that end, for every $i\in\{1,2\}$ we denote by $\cQ_i$ the residual distribution of $\cQ$ corresponding to player $\cP_i$.  Namely, $\cQ_i$ samples $(q_1,q_2)\gets \cQ$ and outputs $q_i$.  Similarly, we denote by $\cQ|q_1$ to be the distribution that samples $(q'_{1},q'_{2})\gets \cQ$ conditioned on $q'_{1}=q_1$, and outputs $q'_{2}$. 

We next define $(\cP_1,\cP_2)$:
\begin{enumerate}
\item\label{item:step1} Choose $q'_1\gets \cQ_1$ and generate  $\sk\gets\Gen(1^\secp)$. Let $\ct_1\gets \Enc(\sk,q'_1)$ and $\ct_2=\tilde{\cP}(\ct_1)$ (i.e., $\ct_2$ is the first message sent by $\tilde{\cP}$ upon receiving $\ct_1$ from the verifier).
\item $\cP_2$ has the ciphertexts~$(\ct_1,\ct_2)$  hardwired into it. On input~$q_2$, it simply emulates the response of $\tilde{\cP}$ given the first three messages $(\ct_1,\ct_2,q_2)$ to obtain~$a_2$.  It outputs $a_2$.

\item $\cP_1$ also has ~$(\ct_1,\ct_2)$  hardwired into it. On input~$q_1$, it computes and outputs~$a_1$ that maximizes the probability of the verifier accepting (w.r.t.\  $\cP_2$ defined above).
Namely, it outputs
$$
a_1={\sf argmax}_{a_1}\Pr_{q_2\leftarrow\cQ|q_1}[\cV(q_1,q_2,a_1,\cP_2(q_2))=1].
$$
\end{enumerate}


We next argue that there exists a negligible function~$\mu=\mu(\secp)$ such that for every $\secp\in\mathbb{N}$, 
$$
\Pr[(\cP_1(\secp),\cP_2(\secp),\cV)=1]\geq \v'-\mu(\secp),
$$
as desired.
To this end, suppose for the sake of contradiction that there exists a non-negligible $\epsilon=\epsilon(\secp)$ such that for every $\secp\in\mathbb{N}$,
\begin{equation}\label{eqn:cont}
\Pr[(\cP_1(\secp),\cP_2(\secp),\cV)=1]\leq \v'-\epsilon(\secp).
\end{equation}

\lisa{here}

We construct an adversary $\cA$ of size $2^{|q_1|+|a_1|}\cdot\poly(\secp/\epsilon)$ that breaks the semantic security of the underlying encryption scheme with advantage $\frac{\epsilon}{4}$.

The adversary~$\cA$ will use his challenge ciphertext~$\ct$ to define $\cP_1$ and $\cP_2$, and will use $\cP_1$ and $\cP_2$ in his attack.  Note that $\cP_2$ can be efficiently emulated in time $\poly(\secp)$ (assuming $\secp$ is larger than the communication complexity of $\cG$).  However, $\cP_1$ may not be efficient. In what follows we show that the maximization problem implicit in $\cP_1$ can be \emph{approximated} in (non-uniform) time $2^{|q_1|+|a_1|}\cdot\poly(\secp/\epsilon)$.  More specifically,  we show that there exists a function~$F$, that takes as input a ciphertext $\ct$ and a query $q_1\in {\sf Support}(\cQ_1)$, 
it runs in time  $2^{|a_1|+|q_1|}\cdot {\poly(\frac{\secp}{\epsilon})}$, and for every $q_1\in {\sf Support}(\cQ_1)$,
\begin{equation}\label{eqn:F}
\Pr_{q_2\gets\cQ|q_1}[\cV(q_1,q_2,\cP_1(q_1),\cP_2(q_2))=1]-\Pr_{q_2\gets\cQ|q_1}[\cV(q_1,q_2,F(\ct,q_1),\cP_2(q_2))=1]\leq \frac{\epsilon}{2},
\end{equation}
where $\cP_1,\cP_2$ are defined w.r.t.\ the ciphertext $\ct$.

In what follows, we describe $F$ as having randomized advice, but we will later set its advice to be ``the best possible", and thus obtain a deterministic function.  For every possible $q_1\in{\sf Support}(\cQ_1)$, we hardwire $N=\frac{9(\secp +  |a_1|)}{\epsilon^2}$ queries $q_{2,1},\ldots,q_{2,N}$ sampled independently from the distribution $\cQ|q_1$.

$F(\ct,q_1)$ is computed by approximating for every $a_1$ the probability 
$$
p_{q_1,a_1}=\Pr_{q_2\gets\cQ|q_1}[\cV(q_1,q_2,a_1,\cP_2(q_2))=1]
$$
by its empirical value
$$
p'_{q_1,a_1} = \frac{1}{N}
|\{i: \cV(q_1,q_{2,i},a_1,\cP_2(q_{2,i})=1\}|.
$$
It outputs $a_1$ with the maximal value of $p'_{q_1,a_1}$.

Note that (as a circuit) $F$ is of size $2^{|q_1|+|a_1|}\cdot\poly(\secp)$, since it has hardwired into it $N\cdot 2^{|q_1|}$ queries hardwired ($N$ for each possible $q_1$), and on input $(\ct,q_1)$ it runs in time $2^{|a_1|}\cdot N \cdot\poly(\secp)$.  Thus, its total size is as desired.

By a Chernoff bound,\footnote{\label{footnote:Chernoof}The form of Chernoff bound that we use here is that for $X_1,\ldots,X_N$ identically and independently distributed in $\{0,1\}$ with expectation $\mu$, it holds that $\Pr[|\frac{1}{N}\sum_{i=1}^N X_i -\mu|>\delta]\leq 2^{-2N\delta^2}$.  }
$$
\Pr[|p'_{q_1,a_1}-p_{q_1,a_1}|>\frac{\epsilon}{3}]\leq  2^{-2N({\epsilon}/{3})^2}=2^{-\secp} \cdot 2^{-|a_1|}
$$
From the equation above (and applying a union bound over all $a_1$) indeed the difference between the two probabilities in Equation~\eqref{eqn:F} is at most $\frac{\epsilon}{3}+2^{-\secp}\leq \frac{\epsilon}{2}$, as desired.

We are now ready to define our adversary $\cA$ that will use $(F, P_2)$ to break semantic security.  Specifically, $\cA$ takes as input a tuple $(q_{1,0},q_{2,0},q_{1,1},q_{2,1},\ct)$, where $(q_{1,0},q_{2,0}),(q_{1,1},q_{2,1})\gets\cQ$, and $\ct$ is distributed by choosing $\sk\gets\Gen(1^\secp)$ and $b^*\gets \{0,1\}$, and sampling $\ct\gets \Enc(\sk,q_{1,b^*})$.  
It guesses $b^*$ as follows:

\begin{enumerate}
\item For every $b\in\{0,1\}$, run $F(\ct, q_{1,b})$ in time $2^{|q_1|+|a_1|}\cdot\poly(\secp/\epsilon)$ and compute $a_{1,b}=F(\ct,q_{1,b})$. 
\item If there exists $b\in\{0,1\}$ such that 
$$
\cV(q_{1,b},q_{2,b},a_{1,b},\cP_2(q_{2,b}))=1~~\wedge~~\cV(q_{1,1-b},q_{2,1-b},a_{1,1-b},\cP_2(q_{2,1-b}))=0
$$ 
then output $b$.  Else output a random $b\gets \{0,1\}$.
\end{enumerate}  
Note that by definition of $\cP^*$ and $(\cP_1,\cP_2)$, it holds that for $b=b^*$,
$$
\Pr_{(q_{1,b},q_{2,b})\leftarrow \cQ}[\cV(q_{1,b},q_{2,b},\cP_1(q_{b}),\cP_2(q_{2,b}))=1]\geq v'.
$$
Thus, by Equation~\eqref{eqn:F}, it holds that for $b=b^*$,
\begin{equation}\label{eqn:v'} 
 \Pr_{(q_{1,b},q_{2,b})\leftarrow \cQ}[\cV(q_{1,b},q_{2,b},F(\ct,q_{1,b}),\cP_2(q_{2,b}))=1]\geq v'-\frac{\epsilon}{2}.
\end{equation}
On the other hand, by our contradiction assumption (Equation~\eqref{eqn:cont}) it holds that for $b=1-b^*$, 
\begin{equation}\label{eqn:v'-eps}
 \Pr_{(q_{1,b},q_{2,b})\leftarrow \cQ}[\cV(q_{1,b},q_{2,b},F(\ct,q_{1,b}),\cP_2(q_{2,b}))=1]\leq v'-\epsilon.
\end{equation}



Denote by $E_\good$ the event that
$$\cV(q_{1,b},q_{2,b},F(\ct,q_{1,b}),\cP_2(q_{2,b}))=1~~~~\mbox{ for }~~~b=b^*
$$
and 
$$\cV(q_{1,b},q_{2,b},F(\ct,q_{1,b}),\cP_2(q_{2,b}))=0~~~~\mbox{ for }~~~b=1-b^*
$$ 
Similarly, denote by $E_\bad$ the event that 
$$\cV(q_{1,b},q_{2,b},F(\ct,q_{1,b}),\cP_2(q_{2,b}))=0~~~~\mbox{ for }~~~b=b^*
$$
and 
$$\cV(q_{1,b},q_{2,b},F(\ct,q_{1,b}),\cP_2(q_{2,b}))=1~~~~\mbox{ for }~~~b=1-b^* 
$$
Denote by $E$  the event that 
$$\cV(q_{1,b},q_{2,b},F(\ct,q_{1,b}),\cP_2(q_{2,b}))=1~~~~\forall b\in\{0,1\}
$$
Then by Equation~\eqref{eqn:v'},
$$
\Pr[E_\good]+\Pr[E]
\geq v'-\frac{\epsilon}{2},
$$
and by Equation~\eqref{eqn:v'-eps},
$$
\Pr[E_\bad]+\Pr[E]\leq v'-\epsilon,
$$
which together imply that 
\begin{equation}\label{eqn:eps}
\Pr[E_\good]-\Pr[E_\bad]\geq\frac{\epsilon}{2}.
\end{equation}
Thus,
 \begin{align*}
&\Pr[b=b^*]\geq\\
&\frac{1}{2}\cdot (1-\Pr[E_\good\cup E_\bad])+\Pr[E_\good]\geq\\
&\frac{1}{2}+\frac{1}{2}(\Pr[E_\good]-\Pr[E_\bad])=\\
&\frac{1}{2}+\frac{\epsilon}{4},
\end{align*}
as desired, where the first equation follows from the definition of $E_\good$ and $E_\bad$ and the definition of $\cA$, the second equation follows from the union bound, the third equation follows from Equation~\eqref{eqn:eps}. This contradicts the security of $\LQHE$; thus, we conclude the desired bound on the classical $\mathsf{CS}$ value of $\cT^{\cG}$. \qedhere

\end{enumerate}
\end{proof}


\subsection{Extension to $k$-Player Games}\label{sec:k-provers}
 In this section, we generalize \cref{thm:main} to $k$-player games for $k>2$. We begin with a construction that is a $2k$-round analogue of the transformation $\cT$ from \cref{thm:main}: given any $k$-player non-local game $\cG$, we define the following interactive game $\cT^{\cG}$:
 
 \begin{enumerate}
     \item The verifier samples $(q_1, \hdots, q_k) \gets \cQ$, $\sk_1, \hdots, \sk_{k-1} \gets \Gen(1^\secp)$, and $\hat q_i \gets \Enc(\sk_i, q_i)$ for each $1\leq i \leq k-1$. 
     \item For each $1\leq i\leq k-1$, in round $2i-1$ the verifier sends $\hat q_i$. In round $2i$ the prover responds with a ciphertext $\hat a_i$.
     \item In round $2k-1$ the verifier sends $q_k$; in round $2k$ the prover responds with some string $a_k$.
     \item The verifier decrypts each $\hat a_i$ with $\sk_i$ and accepts if and only if the transcript $(q_1, a_1, \hdots, q_k, a_k)$ is accepting according to $\cG$.
 \end{enumerate}
 
 We prove the following theorem.
 
 \begin{theorem}\label{thm:main-k-prover}
Fix any $\LQHE$ scheme (satisfying correctness with respect to auxiliary inputs) for a circuit class $\cC$, and any $k$-player non-local game $\cG=(\cQ,\cV)$ with classical value $v$ and quantum value $\v^*$, such that the value $\v^*$ is obtained by a prover strategy $(C^*_1,\hdots, C^*_k)$ with a quantum state $\ket{\Psi}\in\cH_1\tensor \hdots \tensor \cH_k$ with each $C^*_i\in \cC$ (except possibly $C^*_k$). Denote by $|q_i|$ and $|a_i|$ the lengths of the query and answer of $P_i$, respectively.  If the underlying $\LQHE$ encryption scheme is $T$-secure, for  $T(\secp)=2^{\sum_{i=1}^{k-1} (|q_i|+|a_i|)}\cdot\poly(\secp)$, then the following holds:
\begin{enumerate}
    \item The quantum $\CS$ value of $\cT^{\cG}$ is at least $\v^*$. 
    \item The classical $\CS$ value of $\cT^{\cG}$ is at most $\v+\negl(\secp)$.
\end{enumerate}
\end{theorem}

\begin{proof}
We briefly sketch the quantum $\mathsf{CS}$ value of $\cT^{\cG}$. Given a $k$-tuple of entangled provers $\cP_1, \hdots \cP_k$ (with shared state $\ket{\Psi}_{\RegA_1, \hdots, \RegA_k}$), we define the following prover $\cP$ for the interactive game:

\begin{itemize}
    \item $\cP$ initially has internal state $\ket{\Psi}_{\RegA_1, \hdots, \RegA_k}$.
    \item Given $\hat q_i$ (for each $1\leq i\leq k-1$), $\cP$ homomorphically evaluates the circuit defining $\cP_i$ on register $\RegA_i$ and $\hat q_i$ (tracing out any ancilla registers). $\cP$ sends the encrypted answer $\hat a_i$ to the verifier.
    \item Given $q_k$, $\cP$ evaluates the circuit defining $\cP_k$ on $\RegA_k$ (and $q_k$), and sends the answer $a_k$ to the verifier.
\end{itemize}
Analogously to \cref{thm:main}, the aux-input correctness of $\LQHE$ implies that the verifier will accept with probability $v^*(\cP_1, \hdots, \cP_k) \pm \negl(\secp)$, where $v^*(\cP_1, \hdots, \cP_k, \ket{\Psi})$ denotes the value of strategy $(\cP_1, \hdots, \cP_k, \ket{\Psi})$. In more detail, we invoke aux-input correctness and the contractivity of trace distance $k-1$ times sequentially (starting with auxiliary registers $(\RegA_2, \hdots, \RegA_k)$ and removing one $\RegA_i$ each time).

We now bound the classical value of $\cT^\cG$ via the following argument. Suppose that a (computationally bounded) classical interactive prover $\widetilde{\cP}$ (deterministic without loss of generality) has value $v'$ in $\cT^{\cG}$. We will construct local provers $(P^*_1, \hdots, P^*_k)$ winning $\cG$ with probability at least $v' - \negl(\secp)$.

To this end, we sample $(q_1', \hdots, q_k') \gets \cQ$,  secret keys $\sk'_1, \hdots, \sk'_{k-1}\gets\Gen(1^\secp)$ and $\ct'_i \gets \Enc(\sk'_i, q'_i)$ for $1\leq i \leq k-1$. Each prover $P^*_i$ has $\ct'_1,\hdots,\ct'_{i}$ hardwired into its description. 
The prover $P^*_k$ simply emulates the last message function of $\widetilde{P}$; namely, upon receiving $q_k$ it emulates $\widetilde{P}$ assuming that the first $k-1$ messages from the verifier were $\ct'_1,\ldots,\ct'_{k-1}$.
 We next define $P^*_1,\hdots,P^*_{k-1}$ recursively starting with $P^*_1$.  
 
 Assuming we have already defined $P^*_1,\ldots,P^*_{\ell-1}$ (this includes the base case $\ell = 1$), we  define $P^*_{\ell}$ and an {\em interactive} prover $\widetilde{P}_{\ell+1, \hdots, k}$ that has $\ct'_1,\hdots,\ct'_{\ell}$ hardwired to its description, and is sequentially given $\hat q_{\ell+1}, \hdots, \hat q_{k-1}, q_k$ as inputs and returns $\hat a_{\ell+1}, \hdots, \hat a_{k-1}, a_k$ as outputs.



\begin{itemize}
    \item $\widetilde{P}_{\ell+1, \hdots, k}$ simply emulates $\widetilde{P}$ using hard-coded $\ct'_1, \hdots, \ct'_{\ell}$. 
    \item $P^*_{\ell}$  is given as input $q_{\ell}$ and outputs an optimum of the following maximization problem:
    \[a_{\ell}^*=\arg\max_{a_{\ell}}\Pr_{\{q_j\}_{j\neq \ell}\leftarrow\cQ|q_{\ell}}[\cV(q_1,a_1, \hdots q_k, a_k)=1],
    \]
    where $a_j = P^*_j(q_j)$ for all $j < \ell$, and for all $j > \ell$, $a_j$ is obtained by running $\widetilde{P}_{\ell+1, \hdots, k}$ on encryptions (under fresh secret keys) of $q_{\ell+1}, \hdots, q_k$ and then (unless $j=k$) decrypting the resulting answers.
\end{itemize}

Note that by construction, $P^*_k=\widetilde{P}_{k}$, and $P^*_1,\ldots,P^*_k$ are indeed local.
\yael{fix}
Moreover, just as in the proof of \cref{thm:main}, we can \emph{approximately} solve the maximization problems defined in $P^*_1,\ldots P^*_{k-1}$ with functions $F_1, \hdots, F_{k-1}$ that can be implemented in time $2^{\sum_{i=1}^{k-1} |q_i| + |a_i|}\cdot \poly(\secp)$. This is done, given an inverse polynomial error $\epsilon$, by hard-coding for each $q_i$, $N = \frac{18k^2(\secp + |a_i|)}{\epsilon^2}$ samples $\{q_j^{(\ell)}\}_{j\neq k}$ (for $1\leq \ell \leq N$) from $\cQ|_{q_i}$, and will result in provers $F_1, \hdots, F_{k-1}, P^*_k$ that attain value matching $P^*_1, \hdots, P^*_k$ up to error $\epsilon/4$.

Thus, to complete the proof of  \cref{thm:main-k-prover} it remains to prove the following claim.

\begin{claim} The tuple $(P^*_1, \hdots, P^*_k)$ has success probability at least $v' - \negl(\secp)$.
\end{claim}
\begin{proof}
Assume that $(P^*_1, \hdots, P^*_k)$ has success probability at most $v' - \epsilon$ for some non-negligible $\epsilon$. We first replace $P^*_i$ by $F_i$ defined above, and obtain that $(F_1, \hdots, F_k)$ has success probability at most $v' - 3\epsilon/4$. We now derive a contradiction by a hybrid argument. Specifically, for every $j$, we define the quantity 
\[\mathsf{Hyb}_j = \underset{\substack{q_1, \hdots, q_k \gets \cQ \\ \text{ for } i \leq j: \hspace{.1cm} a_i = F_i(q_i) \\ \text{ for } i > j: \hspace{.1cm} a_i = \Dec(\hat a_i), \hspace{.1cm} \hat a_i \text{ output by } \widetilde{P}_{j+1, \hdots, k}} }{\Pr}\left[ \mathcal V(q_1, a_1, \hdots, q_k, a_k) = 1  \right]
\]
Note that $\mathsf{Hyb}_0$ is equal to the success probaiblity of $\widetilde{P}$, which is equal to $v'$ by assumption, while $\mathsf{Hyb}_{k-1}$ is equal to the value of $(F_1, \hdots, F_k)$. 

We now claim that $\mathsf{Hyb}_j > \mathsf{Hyb}_{j-1} - \frac{\epsilon}{4k} - \negl(\secp)$ for every $j$. To prove this, we will reduce from the security of $\LQHE$ with respect to ciphertext $\ct'_j$; note that $F_1, \hdots, F_{j-1}$ do not depend on $\ct'_j$. Ciphertexts $\ct'_1, \hdots, \ct'_{j-1}$ will remain fixed for this entire argument, while $\ct'_{j+1}, \hdots, \ct'_{k-1}$ are not used by any algorithms in $\mathsf{Hyb}_{j-1}$ or $\mathsf{Hyb}_j$.

Define the auxiliary quantity $\Hyb'_j$ to be the same as $\Hyb_j$, except that $\ct'_j$ is sampled as $\Enc(\sk'_j, q_j)$, where $q_j$ is the input sent to $F_j$ in the experiment. Note that $\Hyb'_j > \Hyb_{j-1} - \frac{\epsilon}{4k}$, because the particular choice of $a_j^* = \Dec(\widetilde{P}_{j, \hdots, k}(\ct'_j))$ in the maximization problem defining $P^*_j$ would have value $\Hyb_{j-1}$ (as this strategy matches the value of $(F_1, \hdots, F_{j-1}, \widetilde{P}_{j, \hdots, k})$), and $F_j$ approximates the $P^*_j$ maximization problem up to error $\frac{\epsilon}{4k}$.

Moreover, it holds that $\Hyb'_j - \Hyb_{j-1} = \negl(\secp)$, or this would result in an efficient test distinguishing encryptions of $q_j$ vs. encryptions of $q_j'$ (by an analogous reduction as in the proof of \cref{thm:main}). 

Thus, by a hybrid argument,\footnote{As discussed in \cite{ePrint:FisMit21} (although context differs slightly here), a hybrid argument can be applied because the collection of indistinguishability claims $\Hyb'_j \approx \Hyb_{j-1}$ are proved via a universal reduction $R$ from the security of $\LQHE$.} we conclude that $\Hyb_{k-1} > \Hyb_0 - \frac{\epsilon}4- \negl(\secp) = v' - \frac{\epsilon}4- \negl(\secp)$, contradicting our initial assumption. This completes the proof of the claim.
\end{proof}
\noindent This completes the proof of \cref{thm:main-k-prover}.
\end{proof}

\subsection{Achieving $1-\negl(\secp)$ Quantum-Classical Gap}
\cref{thm:main,thm:main-k-prover} show how to convert a $k$-prover non-local game $\cG$ into a $2k$-round interactive argument for quantum advantage. The simplest instantiation of this paradigm is to compile a constant-size non-local game $\cG_0$ (such as CHSH, Magic Square, Odd-Cycle Test, GHZ), which will result in a constant-round protocol with constant gap between the quantum and classical values of the game.

We now discuss various methods of obtaining \emph{optimal} quantum-classical gap: namely, quantum value $1-\negl(\secp)$ and classical value $\negl(\secp)$.

\begin{enumerate}
    \item \textbf{Sequential Repetition.} The simplest method is repeating $\cT(\cG_0)$ $\secp$ times in sequence. If the quantum value of $\cT(\cG_0)$ (equivalently, the quantum value of $\cG_0$) is $1-\negl(\secp)$, then this will also hold for a sequential repetition, while if the classical value of $\cT(\cG_0)$ is bounded away from $1$, then the sequentially repeated game's classical value will be $\negl(\secp)$ by a standard argument.
    
    In the case where the quantum value of $\cG_0$ is less than $1$, one can apply \textbf{threshold sequential repetition}, in which the verifier for the repeated game accepts if at least a $\theta$-fraction of the copies of $\cT(\cG_0)$ would be accepted. Choosing $\theta$ strictly between the classical and quantum values of $\cG$ will ``polarize'' the classical and quantum values to $(\negl(\secp), 1-\negl(\secp))$.
    \item \textbf{Random-Terminating Parallel Repetition.} Option 1 increases the round complexity of the interactive game to $O(\secp)$, which is somewhat undesirable. Ideally, it would be possible to execute $\secp$ \emph{parallel} copies of $\cT(\cG_0)$ (preserving the round complexity), but it is known that parallel repetition fails to amplify the (classical) computational values of certain interactive games \cite{FOCS:BelImpNao97}. 
    
    However, by appealing to the (slightly more complex operation of) \emph{random-terminating} parallel repetition \cite{FOCS:Haitner09}, we can again ``polarize'' the classical and quantum values of $\cT(\cG_0)$ (while preserving the round complexity). Again, we consider a threshold variant of random-terminating parallel repetition in which the verifier accepts based on a threshold $\theta$ strictly between the quantum and classical values of $\cG_0$. \cite{FOCS:Haitner09} implies that the classical value of this repeated game will be $\negl(\secp)$ (using $\secp$ repetitions), while the quantum value of this game will be $1-\negl(\secp)$.
    \item \textbf{Plain Parallel Repetition?} Finally, we revisit the question of whether plain (threshold) parallel repetition suffices -- this would result in a simpler protocol as compared to Option 2. A lower bound on the quantum value is immediate; what is unclear is a negligible upper bound on the classical value. However, we observe that (threshold) parallel repetitions of $\cT(\cG_0)$ can be analyzed by appealing to \cref{thm:main,thm:main-k-prover} with respect to (threshold) parallel repetitions of the \emph{non-local game} $\cG_0$. 
    
    In particular, under a sufficiently strong assumption on the $\LQHE$, the classical value of this repeated game is at most the classical value of the (threshold) parallel repetition of $\cG_0$ as a non-local game! Thus, if sufficient (threshold) parallel repetition of $\cG_0$ results in a game with negligible classical value, then under a sub-exponential hardness assumption, the parallel repeated interactive game has $\negl(\secp)$ classical value. 
    
    Since, for example, such parallel repetition theorems are known for $2$-player games \cite{STOC:Raz95,STOC:Rao08}, this resolves the question for $2$-player games under a subexponential assumption. We elaborate below on how to avoid this subexponential loss for games $\cG_0$ that exhibit \emph{exponential} hardness amplification under parallel repetition.
\end{enumerate}

\subsubsection{Avoiding Sub-exponential Loss for Games with Strong Parallel Repetition}\label{subsec:poly-assumptions}

We briefly recall the formal definition of \emph{(threshold) parallel repeated games} $\cG = \cG_0^{t, \theta}$:

\begin{definition}[Threshold Parallel Repetition]\label{def:threshold-repeat}
  Let $\cG_0$ denote a constant-size $k$-prover non-local game. We define the $t$-fold $\theta$-threshold repeated game $\cG = \cG_0^{t, \theta}$ as follows:
  
  \begin{itemize}
      \item Queries are of the form $(\tilde q_1 = (q_{1,1}, \hdots, q_{1, t}), \hdots, \tilde q_k = (q_{k,1}, \hdots, q_{k, t}))$. The $k$-tuples $(q_{1, i}, \hdots, q_{k, i})$ are sampled i.i.d.
      \item Answers have the form $(\tilde a_1 = (a_{1,1}, \hdots, a_{1, t}), \hdots, \tilde a_k = (a_{k, 1}, \hdots, a_{k,t}))$.
      \item The repeated verifier $\cV$ accepts if at least $\theta$ fraction of the transcripts $(q_{1,i}, a_{1,i}, \hdots, q_{k,i}, a_{k,i})$ are accepted by the $\cG_0$-verifier $\cV_0$. 
  \end{itemize}
\end{definition}
  
\begin{remark}\label{remark:quantum-value-rep} We remark that if the quantum value of $\cG_0$ is at least $v^* > \theta + \epsilon$, then by a Chernoff bound (see Footnote~\ref{footnote:Chernoof}). the quantum value of the repeated game $\cG$ is at least
  \[ 1 - 2^{2t \epsilon^2}.
  \]  
  On the other hand, \textbf{\emph{if $k=2$ and}} the classical value of $\cG_0$ is \emph{at most} $v < \theta - \epsilon$, then by \cite{STOC:Rao08}, the classical value of $\cG$ is at most
 \[ 2^{-\gamma \epsilon^3 t / |a_1|}, 
 \]
 where $\gamma$ is a constant that can depend on $\theta$. Thus, for large enough $t = O(\secp)$, $\cG$ will have quantum value $\geq 1-2^{-\secp}$ and classical value $\leq 2^{-\secp}$ in the case $k=2$.
 
 For general $k$, the status of parallel repeated games is considerably less well understood \cite{Verbitsky96,arxiv:DHVY16,arxiv:HolRaz20}. 
 \end{remark}
 
 In this section, we prove a strengthening of \cref{thm:main} for all games $\cG_0$ with strong enough parallel repetition properties:
 
 \begin{theorem}\label{thm:poly-hardness} Let $\cG_0$ be a $k$-player non-local game of constant size; i.e., the lengths of queries and answers in $\cG_0$ is $O(1)$. Let $v$ be the classical value of $\cG_)$ and let $v^*$ be its quantum value, and suppose that $v^*>v$. 
   Let $\cG_0^{t, \theta}$ denote the threshold parallel repetition of~$\cG_0$ (\Cref{def:threshold-repeat}), and let $\theta = \frac{v + v^*}2$.   
   
   \emph{Assume that for all $t$}, $v(\cG_0^{t, \theta}) = 2^{-c t}$ for a fixed constant $c$ (that can depend on $\cG_0$).
   
   Then, given a \emph{polynomially secure} QHE scheme $\LQHE$, the game $\cG = \cG_0^{\secp, \theta}$ can be converted into a single-prover game that has (computational) quantum value $1 - \negl(\secp)$ and (computational) classical value $\negl(\secp)$.
 \end{theorem}
 For example, using \cref{thm:poly-hardness}, we can compile arbitrary parallel repeated (constant-size) $2$-prover games under polynomial hardness assumptions.
 
 \begin{proof}
   We consider the single prover protocol defined in \cref{thm:main-k-prover} for the game $\cG$ with the following modification:\footnote{If the QHE scheme is \emph{public-key}, this modification is unnecessary.} we sample a fresh QHE key for each parallel instance of $\cG_0$ (this simply corresponds to a parallel repetition of the $\cG_0$-compiled protocol as defined in \cref{thm:main-k-prover}). 
   
   The fact that the quantum value is $1-\negl(\secp)$ follows from the fact that the quantum value of $\cG$ is $1-\negl(\secp)$ (as observed in \Cref{remark:quantum-value-rep}) together with \cref{thm:main-k-prover}. We focus on bounding the classical value.
   
   For the sake of contradiction, assume that $\widetilde{\cP}$ is a poly-size prover for the compiled game attaining value at least $\delta$ for non-negligible $\delta = \delta(\secp)$. This means that with non-negligible probability, $\widetilde{\cP}$ (implicitly) produces transcripts $(\tilde q_1, \tilde a_1, \hdots, \tilde q_k, \tilde a_k)$ such that at least $\theta$ fraction of the $\cG_0$-transcripts $(q_{1,i}, a_{1,i},\hdots,  q_{k,i}, a_{k,i})$ are accepting. 
   
   Let $t' = O(\log \secp)$ be defined so that the classical value of $G_0^{t', \theta} < \delta/2$ (such $t= O(\log \secp)$ exists by our assumption on $\cG_0$).
   ~By an averaging argument, we know that there exists a subset $S\subset [t]$ with $|S| = t'$ and queries $\{q_{1,i}, \hdots, q_{k, i}\}_{i\not\in S}$ such that $\widetilde{\cP}$ wins with probability at least $\delta$ conditioned on $\{q_{j,i}\}_{i\not\in S}$.
   
   This allows us to construct a prover $\widetilde{\cP'}$ for the $\LQHE$-compiled variant of $\cG^{t', \theta}$ that wins with probability $\delta$: $\widetilde{\cP'}$ simply has $S, \{q_{j,i}\}_{i\not\in S}$ hardcoded and emulates $\widetilde{P}$ by sampling the $([t]-S)$-slot messages itself (this requires either $\LQHE$ to be public-key or for the protocol to use independent secret keys for the different slots).
   
   Finally, we see that since the classical value of $\cG^{t', \theta}$ is at most $\delta/2$ but $\widetilde{\cP'}$ wins the interactive game with probability at least $\delta$, $\widetilde{\cP'}$ contradicts \cref{thm:main-k-prover} assuming the polynomial security of $\LQHE$. 
 \end{proof}
 
 \subsection{Non-Interactive Protocols in the ROM}
 We briefly remark on the ability to convert protocols arising from \cref{thm:main} into non-interactive protocols for verifying quantum advantage.
 
 Specifically, suppose that $\cG$ is a $2$-player game such that the distribution of $q_2$ is \emph{uniform and independent of $q_1$}. Then, we know (by \cref{thm:main,thm:poly-hardness}) that given an $\LQHE$ scheme, we can convert $\cG$ into a $4$-message interactive game with quantum value $1-\negl(\secp)$ and classical value $\negl(\secp)$. Moreover, the $3$rd message of this interactive game is \emph{public-coin}. Therefore, we can apply the Fiat-Shamir heuristic \cite{C:FiaSha86} to this interactive game to obtain a \emph{non-interactive game} with quantum value $1-\negl(\secp)$ and classical value $\negl(\secp)$ in the random oracle model.\footnote{Our interactive game is privately verifiable, but the classical soundness reduction for Fiat-Shamir extends immediately to this case.} This template can be instantiated using (for example) the CHSH game or the magic square game to obtain new $2$-message quantum advantage protocols in the random oracle model.
 
 We observe that even the ``honest'' quantum prover (attaining $1-\negl(\secp)$ value) in this non-interactive game only requires classical access to the random oracle (to hash its own classical message).

%% file: 4-quantum-advantage.tex
\section{Protocols for Verifying Quantum Advantage}

In this section, we give a concrete instantiation of our framework and outline directions for future work, focusing on obtaining protocols with simple(r) quantum provers. 
~We proceed to describe a concrete instantiation of our blueprint, using the CHSH game and the Mahadev QFHE scheme.




\subsection{Compiling the CHSH Game}\label{sec:CHSH}
We first recall the CHSH game.

\newcommand{\CHSH}{\mathsf{CHSH}}
\begin{definition}[The CHSH Game]\label{def:CHSH}
The CHSH game $\cG_\CHSH$ consists of the uniform query distribution $\cQ_\CHSH$ over $(q_1,q_2)\in(\{0,1\})^2$ and verification predicate $\cV_\CHSH(q_1,q_2,a_1,a_2)$ which is $1$ if and only if
$$a_1 \oplus a_2 = q_1 q_2 \pmod{2}.$$

\noindent
The classical value of this game is $v_\CHSH=0.75$ and the quantum value is $v^*_\CHSH=\cos^2(\pi/8)\approx 0.85$.
The optimal quantum strategy is as follows:
The two players $\cA,\cB$ share an EPR pair 
$$\frac{1}{\sqrt{2}} (\ket{0}_{\cA}\ket{0}_{\cB} + \ket{1}_{\cA}\ket{1}_{\cB})$$
where players $\cA,\cB$ have the $\cA,\cB$ registers respectively.
Upon receiving $q_1$, player $\cA$ measures her register $\cA$ in the Hadamard ($\pi/4$) basis if $q_1=0$, and in the standard basis if $q_1=1$, and reports the outcome $a_1\in\{0,1\}$. Player $\cB$ measures her register $\cB$ in the $\pi/8$-basis if $q_2 = 0$ and in the $3\pi/8$-basis if $q_2 = 1$, and reports the outcome $a_2 \in \{0,1\}$.
\end{definition}

\noindent
By Theorem~\ref{thm:main} and ~\cref{claim:QFHE}, we have the following corollary:
\begin{corollary}\label{cor:CHSH}
Consider the \cite{Mah18a} QHE scheme for poly-size circuits in the Toffoli and Clifford basis. Consider the CHSH game $\cG_\CHSH=(\cQ_\CHSH,\cV_\CHSH)$ and quantum strategy in Definition~\ref{def:CHSH} where $$\ket{\Psi}=\frac{1}{\sqrt{2}} (\ket{0}_{\cA}\ket{0}_{\cB} + \ket{1}_{\cA}\ket{1}_{\cB})$$ 
and $C^*_1(\ket{\Psi}_\cA,\cdot)\in\cC$. 

The single-player computationally sound \ig~$\cT^\cG$ has:
\begin{itemize}
    \item quantum $\CS$ value $= \v^*_\CHSH \ge 0.85$
    \item classical $\CS$ value $\geq \v_\CHSH+\negl(\secp) = 0.75+\negl(\secp)$.
\end{itemize}
\end{corollary}
Amplifying this gap can be done by sequential repetition. Alternatively, one can compile a parallel-repeated version of the CHSH game to get a protocol with a large gap between the quantum and classical $\CS$ values.




\paragraph{Prover Efficiency.} The compilation of the CHSH game with Mahadev's QFHE (Corollary~\ref{cor:CHSH}) results in a conceptually simple 4-round protocol with a relatively simple quantum prover. Here, we analyze the quantum prover's algorithm.

Returning to the quantum strategy for the CHSH game (Definition~\ref{def:CHSH}), player $\cA$ applies a controlled Hadamard gate to $\ket{\bar{q_1}}\ket{\Psi}_\cA$, \yael{Did you mean $q_1$ instead of $\bar{q_1}$?} and then measures the $\cA$ register. This can be implemented by a circuit $C^*_1$ containing Clifford gates and a single Toffoli gate~\cite{bera2008universal}. Recall that in Mahadev's scheme (\cite{Mah18a}), evaluating Clifford gates only requires applying the intended Clifford gate to a (Pauli one-time-padded) encryption of the underlying state. To evaluate a Toffoli gate, the Toffoli gate is applied to the encrypted qubit, followed by 3 ``encrypted CNOT'' operations and 2 Hadamard gates. The bulk of the computational cost of the prover is in the encrypted CNOT operations. Using the trapdoor claw-free functions (TCF) for the classical ciphertexts in the Mahadev QHE scheme, an encrypted CNOT operation consists of creating a uniform superposition over the TCF domain, evaluating the function in superposition, measurements and Clifford gates. That is, this requires $p(\secp):=\log|\cD|+\log|\cR|$ ancilla qubits corresponding to the TCF domain $\cD$ and range $\cR$.
Concretely, the prover's quantum operations in our compilation of CHSH are as follows.
\begin{itemize}
    \item The prover creates an EPR pair which involves applying a Hadamard and a CNOT gate.
    \item The prover receives a classical ciphertext $\hat{q_1}$ from the verifier in round 1.
    \item The prover homomorphically evaluates $C^*_1(\ket{\Psi}_\cA,q_1)$ (where $C^*_1$ implements player $\cA$'s strategy in CHSH).
    This uses the constant number of qubits in $C^*_1$ and $3p(\secp)$ ancilla qubits for TCF evaluations.
    All of its operations are Clifford gates except a single Toffoli gate and 3 invocations of the TCF evaluation algorithm run in superposition.
    \item The prover sends back a classical ciphertext, and receives a bit $q_2\in\zo$ in round 2.
    It measures $\ket{\Psi}_\cB$ in the $\pi/8$ or $3\pi/8$ basis (depending on $q_2$, as per player $\cB$'s strategy in CHSH).
    \lisa{(after submission): It seems the prover only needs to hold 1 qubit in round 2? which should be good for maintaining `coherence'}
    It sends back the result as $a_2\in\zo$. In particular, the prover can discard the remaining qubits right after it computes and sends its message in round 1.
\end{itemize}

Overall, the prover uses $3p(\secp)+O(1)$ qubits, and the complexity of its operations is dominated by the 3 TCF evaluations. We note that designing a more efficient QHE scheme supporting the controlled-Hadamard gate, potentially based on simpler assumptions, is an attractive route to improving the prover efficiency. 

\lisa{(after submission): if keep Mahadev's as the QHE, may want to figure out controlled H gate using Toffoli, S, H gates.}

\subsection{Future Directions}\label{sec:future-work}
Our work suggests several directions for future work.

\paragraph{Protocols for quantum advantage with very simple quantum provers?}
First, could there be a non-local game where, once an appropriate bipartite state is set up, one of the prover strategies can be implemented using only Clifford gates? If this were possible, the complexity of the QHE evaluation reduces drastically. In particular, quantum homomorphic encryption schemes handling only Clifford gates are much simpler than general quantum homomorphic encryption schemes (e.g. they only require applying the intended Clifford gates and additional classical computations). This could lead us to a truly efficient protocol that does not require maintaining superpositions with security parameter number of qubits. On the other hand, if it were to be the case that any non-local game requires {\em both} provers to be non-Clifford, that would be an interesting outcome as well. To the best of our knowledge, this statement is not known, and does not seem to follow from Gottesman-Knill-type classical simulation of Clifford circuits. 

\paragraph{Simpler homomorphic encryption schemes.} Can we design better somewhat homomorphic encryption schemes for the compilation in Theorem~\ref{thm:main}? We note that the scheme only needs to support the evaluation of \emph{one} of the two provers' strategies in the non-local game, which can often be implemented by a simple circuit. For example,  designing a scheme that simply supports the controlled-Hadamard gate would suffice for compiling the CHSH game. This may give a quantum advantage protocol with a simpler prover strategy.

\paragraph{Protocols based on different assumptions.} We note that for quantum advantage, we only need soundness against classical adversaries. Namely, the homomorphic encryption should have completeness for some quantum gates, but {\em only needs soundness against classical polynomial-time adversaries}. This opens the door to designing QHEs based on e.g. discrete log style assumptions, or the hardness of factoring, rather than learning with errors. (Indeed, quantum advantage protocols have been constructed using factoring-based TCFs; see \cite{KCVY21}). 

\paragraph{Understanding existing protocols.} The existing interactive protocols for quantum advantage (to our knowledge, $\cite{BCMVV18,BKVV20,KCVY21}$) are presented as an all-in-one package. Intuitively, a protocol testing quantumness should have a component testing for quantum resources, e.g. a test of entanglement, and a component that tests computational power, i.e. the cryptography. Can existing protocols be disentangled to two such components? 
A starting point is \cite{KCVY21} which has some resemblance to our CHSH compilation in Section~\ref{sec:CHSH}, although it does incur two more rounds.

More ambitiously, could we understand any single-prover quantum advantage protocol as compiling a (perhaps contrived) $k$-player non-local game via a somewhat homomorphic encryption scheme? We leave it open to understand the reach and the universality of our framework for constructing protocols for quantum advantage.


%% file: appendix.tex
\appendix 

\section{QHE Correctness with Respect to Auxiliary Input}
\label{app:locality}

In this section, we discuss the relationship between aux-input QHE correctness (\cref{def:QHE-aux}) and prior definitions/schemes \cite{Mah18a,Bra18}. In particular, we sketch a proof of \cref{claim:QFHE}.

We first recall the formal QFHE definitions from \cite{Bra18}. Specifically, we now generalize QFHE to allow for encryption/decryption of \emph{states} and homomorphic evaluation of circuit with quantum (rather than classical) output. However, we impose (as done in \cite{Bra18}) the constraint that encryption and decryption are (qu)bit-by-(qu)bit.

\begin{definition}[Quantum qubit-by-qubit Homomorphic Encryption (bQHE)]\label{def:QHE-bit}
A quantum (qu)bit-by-(qu)bit homomorphic encryption scheme $\bQHE=(\Gen,\Enc,\Eval,\Dec)$ for a class of quantum circuits $\cC$ is a tuple of algorithms with the following syntax:
\begin{itemize}
    \item $\Gen$ is a $\PPT$ algorithm that takes as input the security parameter $1^\secp$ and outputs a (classical) secret key $\sk$ of $\poly(\secp)$ bits;
    \item $\Enc$ is a $\QPT$ algorithm that takes as input a secret key $\sk$ and a qubit $b$, and outputs a ciphertext $\ct$. Additionally, if $b$ is a classical bit, the encryption algorithm is $\PPT$ and the ciphertext $\ct$ is classical;
    \item $\Eval$ is a $\QPT$ algorithm that takes as input a tuple $(\C,\ket{\Psi},\vec{\ct}_{\mathrm{in}})$, where $\C:\cH\times(\mathbb{C}^2)^{\tensor n}\rightarrow (\mathbb{C}^2)^{\tensor m}$ is a quantum circuit, $\ket{\Psi}\in\cH$ is an auxiliary quantum state, and $\vec{\ct}_{\mathrm{in}} = (\ct_1,\ldots,\ct_n)$ is a tuple of $n$ ciphertexts. 
    $\Eval$ computes a quantum circuit
    $\Eval_C(\ket{\Psi}\tensor \ket{0}^{\poly(\secp, n)},\vec{\ct}_{\mathrm{in}})$ which outputs a tuple of ciphertexts $\vec{\ct}_{\mathrm{out}}$.
    \vinod{The only reason $\ket{\Psi}$ exists here is because we want eval to be able to take unencrypted inputs.}

    \item $\Dec$ is a $\PT$ algorithm that takes as input a secret key $\sk$ and ciphertext $\ct$, and outputs a qubit $b$. Additionally, if $\ct$ is a classical ciphertext, the decryption algorithm outputs a bit $b$.
\end{itemize}

In addition to security (which again we can restrict to hold against classical adversaries), we require (following \cite{Bra18}) the following \textbf{correctness property}:

\begin{itemize}
    \item For any quantum circuit $C$ and any $n$-qubit state $\ket{\Phi} = \sum \alpha_{i_1, \hdots, i_n} \ket{x_1 \hdots x_n}$, the following two states have negligible trace distance. The first state $\brho_1$ is defined to be the output of $C(\ket{\Phi})$. The second state $\brho_2$ is defined by sampling $\sk \gets \Gen(1^\secp)$, $\ct_i \gets \Enc(\sk, x_i)$ and computing $\Dec(\sk, \Eval(C,\ct_1, \hdots, \ct_n))$. 
\end{itemize}

\end{definition}
The QFHE scheme of \cite{Bra18} is shown (in \cite{Bra18}) to satisfy the definition above. We now discuss mild structural hypotheses under which an encryption scheme satisfying \cref{def:QHE-bit} also satisfies \cref{def:QHE-aux}.

\begin{itemize}
    \item \textbf{Allowing un-encrypted inputs}: this first property is merely syntax that can be added without loss of generality. \cref{def:QHE-bit} does not have \emph{any} un-encrypted state (i.e. $\ket{\Psi})$ in \cref{def:QHE-aux}). This can be easily rectified by encrypting any additional state $\ket{\Psi}$ during QFHE evaluation.
    \item \textbf{Locality}: this is the most important structural property. We say that a QHE scheme is \emph{local} if homomorphic evaluation of a circuit of the form $C_{\RegA}\tensor \Id_{\RegB}$, acting on $\cH \simeq \cH_{\RegA}\tensor \cH_{\RegB}$, is identical to homomorphic evaluation of $C_{\RegA}$ on (encrypted) register $\RegA$ (and acts as identity on the encrypted $\RegB$).
    \item \textbf{Honest Decryption Correctness}: we say that a QHE scheme satisfies honest decryption correctness if $\Dec(\sk, \Enc(\sk, b))$ is the identity map on $b$. In particular, this implies that the equation $\Dec(\sk, \Enc(\sk, b))$ holds in the presence of any auxiliary input.
\end{itemize}

\begin{claim}\label{claim:local-implies-aux}
  Let $\bQHE$ be a (qu)bit-by-(qu)bit QHE scheme satisfying \cref{def:QHE-bit}. Moreover, extend $\bQHE$ to allow un-encrypted inputs, and assume that $\bQHE$ satisfies \textbf{locality} and \textbf{honest decryption correctness}. Then, $\bQHE$ satisfies \cref{def:QHE-aux}.
\end{claim}

\begin{proof}
  We want to show that $\bQHE$ (as described in \cref{claim:local-implies-aux}) satisfies correctness with respect to auxiliary input. Let $\ket{\Psi}_{\RegA \RegB}$ be a bipartite state and $x$ be a (classical) input as in \cref{def:QHE-aux}, and let $C$ denote a circuit taking as input $x$ and a state on register $\RegA$. By the \emph{basic} correctness property of $\bQHE$ with respect to the circuit $C \tensor \Id_{\RegB}$, we know that the state
  \[ \Dec(\sk, \Eval(C\tensor \Id_{\RegB}, \Enc(\sk, x), \Enc(\sk, \ket{\Psi}_{\RegA \RegB}))
  \]
  is negligibly close in trace distance from the state obtained by evaluating $C\tensor \Id_{\RegB}$ on $x, \ket{\Psi}_{\RegA \RegB}$.
  
  Moreover, by the \emph{locality} of $\bQHE$, we know that $\Eval(C\tensor \Id_{\RegB}, \Enc(\sk, x), \Enc(\sk, \ket{\Psi}_{\RegA \RegB}))$ is equivalent to applying $\Eval(C, \cdot)$ to $\Enc(\sk, x)$ and the encrypted $\RegA$ register of $\Enc(\ket{\Psi}_{\RegA \RegB})$. Combining this with the honest decryption correctness property on register $\RegB$, we conclude that the state obtained by evaluating $C\tensor \Id_{\RegB}$ on $x, \ket{\Psi}_{\RegA \RegB}$ is negligibly close in trace distance to the state
  \[ \left(\Dec(\sk, \Eval(C, \Enc(\sk, x), \Enc(\sk, [\ket{\Psi}]_{\RegA})), [\ket{\Psi}]_{\RegB}\right)
  \]
  where $[\ket{\Psi}]_{\RegA}$ and $[\ket{\Psi}]_{\RegB}$ denote the $\RegA$ and $\RegB$ registers of $\ket{\Psi}$, respectively. Finally, the syntax for unencrypted inputs tells us that the above indistinguishability implies the aux-input correctness property of \cref{def:QHE-aux}.
\end{proof}
Thus, to prove \cref{claim:QFHE}, it suffices to observe that the \cite{Mah18a,Bra18} schemes satisfy \textbf{locality} and \textbf{honest decryption correctness}. This can be checked by inspecting \cite{Mah18a,Bra18}; for intuition, we note that any scheme whose homomorphic evaluation is executed in a ``gate-by-gate fashion'' should be local.

\subsection{The \cite{Mah18a} QFHE Scheme}

For the rest of this section, we give a more explicit discussion of the \cite{Mah18a} scheme to sketch why it satisfies our hypotheses. The encryption of a qubit $\ket{\psi}$ in the scheme consists of a quantum one-time padded qubit $X^xZ^z\ket{\psi}$ together with two classical ciphertexts per qubit that encrypt $x$ and $z$. To homomorphically evaluate a gate on some qubits, it applies that gate (plus additional operations if necessary) to those qubits in the padded state, and updates the classical ciphertexts for those qubits. The high-order bit is that when evaluating a circuit $\C$, the qubits which are not acted upon will not be changed, which in turn implies the locality property.


We recall Mahadev's scheme (Scheme 6.1) in more detail (extended slightly to formally allow for encryption of quantum states) and show that it is local. For this discussion, we assume that in a quantum circuit $C$, all measurement gates are deferred to the end.

Recall that $\Gen$ is a key generation algorithm corresponding to a classical homomorphic encryption scheme (that satisfies some desired properties that we elaborate on below).  The encryption of a (classical) message  $M\in\zo^n$ is the one-time pad applied to $M$ together with the classical encryption of the pad.  Using quantum notation, 
$$
\Enc(M)= (X^x\ket{M}, \{\hat{x}_i\}_{i\in[n]})
$$
where $x=(x_1,\ldots,x_n)$ is randomly chosen in $\{0,1\}^n$, and $\{\hat{x}_i\}_{i\in[n]}$ are classical ciphertexts where $\hat{x}_i$ encrypts the pad $x_i$ on the $i$'th qubit (using the underlying classical homomorphic encryption).

The encryption algorithm, $\Enc$, can be extended to encrypt a quantum state $\Psi$ consisting of $\ell$ qubits, in which case, the encryption of $\Psi$ is the quantum one-time pad applied to $\Psi$, together with the classical encryptions of the pad.  Namely,
\begin{equation}\label{eqn:ct}
\Enc(\ket{\Psi})=(Z^zX^x\ket{\Psi},\{\hat{x}_i,\hat{z}_i\}_{i\in[\ell]})
\end{equation}
where $z=(z_1,\ldots,z_\ell)$ and $x=(x_1,\ldots,x_\ell)$ are randomly chosen in $\{0,1\}^\ell$, and $\{\hat{z}_i,\hat{x}_i\}_{i\in[\ell]}$ are classical ciphertexts  where  $\hat{z}_i,\hat{x}_i$ encrypt the Pauli pad on the $i$'th qubit (using the underlying classical homomorphic encryption).

We next focus on the homomorphic operations.  Suppose we wish to homomorphically compute a circuit  $\C:\cH\times\zo^n\to\cH'\times\zo^m$ applied to the quantum state $\ket{\Psi}=\ket{\Psi_0}\ket{M}$.  $\Eval_C$ homomorphically evaluates each gate in $\C$, as described in \cite{Mah18a} (Scheme 6.1), and then measures qubits corresponding to the qubits that $\C$ measures at the end.  
Recall that the input to $\Eval_C$ is  $(\ket{\Psi_0}\ket{\Enc(M)},\ket{0^t})$.  

In what follows, we assume for simplicity that the input to $\Eval_C$ is $(\Enc\ket{\Psi},\ket{0^t})$.  This is without loss of generality since we can think of $\Psi_0$ as being encyrpted with the trivial pad $Z^0X^0$.
Namely, $\Eval_C$ starts with the quantum state $Z^zX^x\ket{\Psi}\ket{0^t}$ and $\ct=\{\hat{x}_i,\hat{z}_i\}_{i\in[\ell]}$ (as defined in Equation~\eqref{eqn:ct}).  It proceeds gate by gate, while satisfying the following desired locality property: For each gate in $\C$, suppose the gate implements the unitary $U\otimes I$, where $U$ acts only on a constant number of registers, denoted by ${\cal J}$, and $I$ is the identity acting on all the other registers, then $\Eval_C$ homomorphically evaluates the gate by applying a unitary $U'\otimes I$ to its padded state, where $U'$ acts on the same set of registers ${\cal J}$ and on the ancilla registers (initialized to be $\ket{0^t}$), and in addition it updates the classical ciphertexts, $\ct$, to encrypt the ``correct pads". 
Note that this locality property implies that scheme is indeed local, as desired.

In what follows we argue that this locality property indeed holds.  It suffices to prove that it holds for Clifford gates and for Toffoli gates (since they form a universal gate set). 

In the following paragraph we use $C$ to denote a Clifford gate (instead of a quantum circuit).  A Clifford gate $C$ has the desired property that it preserves the Pauli group under conjugation, i.e., for every $P_1,P_2$ in the Pauli group there exist $P_3,P_4$ in the Pauli group such that
$$
C(P_1\otimes P_2)C^\dagger=P_3\otimes P_4.
$$
This implies that for every $x,z$ there exists $x',z'$ such that
$$
CZ^z X^x\ket{\Psi}=Z^{z'}X^{x'}C\ket{\Psi}.
$$

Thus, $\Eval_C$ homomorphically computes a Clifford gate, by simply applying the exact same Clifford gate to its padded quantum state, and then obtains a quantum one-time pad of $C\ket{\Psi}$ and homomorphically updates the classical ciphertexts to encrypt $x',z'$. Thus, the locality property is indeed satisfied.

The Toffoli gate, denoted by~$T$, is more complicated, and this is where the additional ancilla registers are used. It cannot be applied to the padded state while only updating the encrypted pads, as was done for Clifford gates, since (as opposed to Clifford gates) it does not
preserve Pauli operators by conjugation.  Applying a Toffoli directly to a 3-qubit one-time padded state results in:
\begin{align*}
&T(Z^{z_1}X^{x_1}\otimes Z^{z_2}X^{x_2}\otimes Z^{z_3}X^{x_3}\ket{\Psi})=\\
&T(Z^{z_1}X^{x_1}\otimes Z^{z_2}X^{x_2}\otimes Z^{z_3}X^{x_3})T^\dagger T\ket{\Psi}=\\
&CNOT^{x_2}_{1,3}CNOT^{x_1}_{2,3}(I\otimes H)CNOT^{z_3}_{1,2}(I\otimes H)(Z^{z_1+x_2z_3}X^{x_1}\otimes Z^{z_2+x_1z_3}X^{x_2}\otimes Z^{z_3}X^{x_1x_2+x_3})T\ket{\Psi}
\end{align*}
where $CNOT^{s}_{i,j}$ is the ``encrypted CNOT operation'' applied to the $i,j$'th qubits, where~$s$ indicates whether to apply the CNOT operation or not, and $H$ is the Hadamard gate (see \cite{Mah18a} for details).  

To homomorphically evaluate a Toffoli gate, $\Eval_C$ applies the Toffoli gate $T$ to the corresponding qubits in its padded state and updates the classical ciphertexts, $\ct$, to be consistent with the pad $$Z^{z_1+x_2z_3}X^{x_1}\otimes Z^{z_2+x_1z_3}X^{x_2}\otimes Z^{z_3}X^{x_1x_2+x_3}
$$
(exactly as was done for Clifford gates).  The remaining piece is to undo the (undesired) operations:  
$$
CNOT^{x_2}_{1,3}CNOT^{x_1}_{2,3}(I\otimes H)CNOT^{z_3}_{1,2}(I\otimes H).$$
Note that undoing the unitary $I\otimes H$ is trivial, since all we need to do is multiply by the conjugate transpose, and this unitary (as well as the unitary corresponding to the Toffoli gate~$T$) satisfies the desired locality property.

Hence, in what follows we focus on how $\Eval$ undoes the encrypted CNOT operations (which is the crux of the difficulty).  Note that $CNOT_{i,j}CNOT_{i,j}$ is the identity, and hence to undo the encrypted CNOT operation $\Eval$ needs to apply $CNOT^s_{i,j}$ given only a ciphertext encrypting the bit $s$. Indeed, the key idea in \cite{Mah18a} is to show how this can be done (locally).  For this, she needs the underlying classical encryption scheme to satisfy certain properties.

Specifically, \cite{Mah18a} uses a special classical encryption scheme that is associated with a trapdoor claw-free function family.\footnote{A trapdoor claw-free function family $\cF$ is a family of injective functions, with a $\PPT$ algorithm that generates a pair of functions $f_0,f_1\in\cF$ together with a trapdoor ${\rm td}$, such that given (a description of) $(f_0,f_1)$   it is hard to find a claw (i.e., a pair $x_0,x_1$ such that $f_0(x_0)=f_1(x_1)$), whereas given $(f_0,f_1)$ together with the trapdoor ${\rm td}$ one can efficiently invert $f_0$ and $f_1$, and in particular, find a claw.}   The encryption scheme has the property that given an encryption of a bit~$s$, denoted by $\hat{s}$, one can efficiently compute a description of a pair of claw free functions $f_0,f_1:\{0,1\}\times {\cal R}\rightarrow {\cal Y}$ such that for every $(\mu_0,r_0),(\mu_1,r_1)\in\{0,1\}\times {\cal R}$ such that $f_0(\mu_0,r_0)=f_1(\mu_1,r_1)$ (a ``claw'') it holds that $\mu_0\oplus \mu_1=s$.  Moreover, given $\hat{s}$ one can efficiently compute the encryption of the trapdoor corresponding to the pair $(f_0,f_1)$.

Armed with this encryption scheme, one can compute the encrypted CNOT operation $CNOT^s$ on a 2-qubit state
$$\ket{\Psi}=\sum_{a,b\in\{0,1\}}\alpha_{a,b}\ket{a,b},$$ 
given $\hat{s}$, as follows:
\begin{enumerate}
\item Classically compute a description of the claw-free pair $(f_0,f_1)$ corresponding to $\hat{s}$.
\item Use the ancilla qubits to entangle $\ket{\Psi}$ with a random claw for $f_0,f_1$, by computing
$$
\sum_{a,b,\mu\in\{0,1\},r\in{\cal R}}\alpha_{a,b}\ket{a,b}\ket{\mu,r}\ket{f_a(\mu,r)}
$$
and measuring the last register to obtain $y\in{\cal Y}$.  Let $(\mu_0,r_0),(\mu_1,r_1)$ be the two preimages of $y$, namely 
$$
f_0(\mu_0,r_0)=f_1(\mu_1,r_1)=y.
$$
Then the remaining state is 
$$
\sum_{a,b\in\{0,1\}}\alpha_{a,b}\ket{a,b}\ket{\mu_a,r_a}.
$$
\item Then XOR $\mu_a$ into the second
register, which results in
\begin{align*}
&\sum_{a,b\in\{0,1\}}\alpha_{a,b}\ket{a,b\oplus \mu_a}\ket{\mu_a,r_a}=\\
&\sum_{a,b\in\{0,1\}}\alpha_{a,b}(I\otimes X^{\mu_0})\ket{a,b+ a\cdot s}\ket{\mu_a,r_a}=\\ &\sum_{a,b\in\{0,1\}}\alpha_{a,b}(I\otimes X^{\mu_0})CNOT^s_{1,2}\ket{a,b}\ket{\mu_a,r_a}
\end{align*}
where the first equation follows from the fact that $\mu_0\oplus \mu_1=s$ and the second equation follows from the definition of $CNOT^s_{1,2}$.
\item Remove the entangled registers $\ket{\mu_a,r_a}$ by applying the Hadamard transform to these registers to obtain
\begin{align*}
\sum_{a,b,d_0\in\{0,1\},d\in\{0,1\}^{\ell}}\alpha_{a,b}(I\otimes X^{\mu_0})CNOT^s_{1,2}\ket{a,b}(-1)^{(d_0,d)\cdot(\mu_a,r_a)}\ket{d_0,d}
\end{align*}
\item Measure the registers $\ket{d_0,d}$ to obtain the state
\begin{align*}
&(I\otimes X^{\mu_0})CNOT^s_{1,2} \sum_{a,b\in\{0,1\}}(-1)^{(d_0,d)\cdot(\mu_a,r_a)}\alpha_{a,b}\ket{a,b}=\\
&(-1)^{(d_0,d)\cdot(\mu_0,r_0)} (I\otimes X^{\mu_0})CNOT^s_{1,2} \left(\sum_{b\in\{0,1\}}\alpha_{0,b}\ket{0,b} + (-1)^{(d_0,d)\cdot((\mu_0,r_0)\oplus(\mu_1,r_1))}\sum_{b\in\{0,1\}}\alpha_{1,b}\ket{1,b}\right)=\\
&(-1)^{(d_0,d)\cdot(\mu_0,r_0)} (I\otimes X^{\mu_0})CNOT^s_{1,2} (Z^{(d_0,d)\cdot((\mu_0,r_0)\oplus(\mu_1,r_1))}\otimes I)\left(\sum_{a,b\in\{0,1\}}\alpha_{a,b}\ket{a,b}\right)=\\
&(-1)^{(d_0,d)\cdot(\mu_0,r_0)} (Z^{(d_0,d)\cdot((\mu_0,r_0)\oplus(\mu_1,r_1))}\otimes X^{\mu_0})CNOT^s_{1,2} \left(\sum_{a,b\in\{0,1\}}\alpha_{a,b}\ket{a,b}\right).
\end{align*}
\item Finally, use the (classical) encryption of the trapdoor for $(f_0,f_1)$ to homomorphically evaluate $(d_0,d)\cdot((\mu_0,r_0)\oplus(\mu_1,r_1))$ and $\mu_0$, and update the classical encryptions of the Pauli pads.
\end{enumerate}
Note that all these operations (in Steps (1)-(6)) are either classical operations relating to updating the (classically) encrypted pad, or quantum operations applied to the 2 qubits of $\Psi$ and to additional ancilla qubits.  By the deferred measurement principle, we can think of the quantum operations as applying a unitary and only then applying the measurement. 

In our context, $\Eval_C$ does these quantum operations on its large quantum state consisting of $\ell+t$ qubits.  However, all these operations (in Steps (1)-(6)) change only the 2-qubits on which the CNOT is applied and some of the ancila registers.
Overall, the unitary that $\Eval_C$ applies, when evaluating $U\otimes I$, is of the form $U'\otimes I$ where the unitary $U'$ is applied only to the qubits that $U$ would act on and to the ancilla qubits, and $I$ is the identity unitary applied to all the other registers, as desired.